\documentclass[12pt]{article}
\usepackage{mwe}
\usepackage{amsthm}
\usepackage{amsmath}
\usepackage{amsfonts}
\usepackage{amssymb}
\usepackage[colorinlistoftodos]{todonotes}
\usepackage{subcaption}
\usepackage{caption}
\usepackage{graphicx}
\usepackage{thmtools}
\usepackage{thm-restate}
\usepackage{hyperref}
\usepackage{cleveref}
\usepackage{enumerate}
\usepackage[shortlabels]{enumitem}

\newtheorem{theorem}{Theorem}
\newtheorem{corollary}[theorem]{Corollary}
\newtheorem{lemma}[theorem]{Lemma}
\newtheorem{claim}[theorem]{Claim}

\newtheorem{observation}[theorem]{Observation}

\newtheoremstyle{case}{}{}{}{}{}{:}{ }{}
\theoremstyle{case}

\newcommand{\rajiv}[1]{\todo[linecolor=red,backgroundcolor=red!25,bordercolor=teal]{RR: #1}}

\newcommand\AG{{\mathcal{A}}(\Gamma)}

\newcommand{\hide}[1]{}
\title{Sweeping Arrangements of Non-Piercing Regions in the Plane}
\author{
Suryendu Dalal \\ IIIT-Delhi, India. \\ suryendud@iiitd.ac.in \and 
Rahul Gangopadhyay\footnote{Research was mostly conducted while Rahul Gangopadhyay was at MIPT and was supported by a grant for research centers in the field
of artificial intelligence, provided by the Analytical Center for the Government of the Russian Federation in accordance with the subsidy agreement (agreement identifier 000000D730321P5Q0002) and the agreement with the Moscow Institute of Physics and Technology dated November 1, 2021 No. 70-2021-00138.} \\ Department of Computer Science and Engineering,\\ Indian Institute of Technology (ISM) Dhanbad,\\ Dhanbad, 826004, Jharkhand, India
, 
\\ rahulg@iiitd.ac.in\and
Rajiv Raman \\ IIIT-Delhi, India. \\ rajiv@iiitd.ac.in \and 
Saurabh Ray \\ NYU Abu Dhabi, UAE. \\ saurabh.ray@nyu.edu}

\begin{document}
\maketitle
\begin{abstract}
Let $\Gamma$ be an arrangement of Jordan curves in the plane, i.e., simple closed curves in the plane. For any curve $\gamma \in \Gamma$, we denote the bounded region enclosed by $\gamma$ as $\tilde{\gamma}$. 
We say that $\Gamma$ is non-piercing if for any two curves $\alpha , \beta \in \Gamma$, $\tilde{\alpha} \,\setminus\, \tilde{\beta}$ is connected.
A non-piercing arrangement of curves generalizes a set of $2$-intersecting curves
in which each pair of curves intersect in at most two points. 
Snoeyink and Hershberger (``Sweeping Arrangements of Curves'', SoCG '89) 
proved that if we are given an arrangement $\Gamma$ of $2$-intersecting curves and a {\em sweep} curve $\gamma\in{\Gamma}$,  then the arrangement can be \emph{swept} by $\gamma$ while always maintaining the $2$-intersecting property of the curves in $\Gamma$.
We generalize the result of Snoeyink and Hershberger to the setting of non-piercing arrangements. 
Given an arrangement $\Gamma$ of non-piercing curves, a sweep curve $\gamma\in \Gamma$, and a point $P$ in $\tilde{\gamma}$, we show that
we can continuously shrink $\gamma$ to $P$ so that throughout the process,
the arrangement remains non-piercing (except at a finite set
of points in time where $\gamma$ crosses other curves), and $P$ lies in $\tilde{\gamma}$. We show that our arguments can be modified
if $P$ lies outside $\tilde{\gamma}$, and we want to sweep $\gamma$ \emph{outwards} so that $P$ lies outside $\tilde{\gamma}$, and
the arrangement remains non-piercing.
As a second contribution, we give an alternate proof of the result of Snoeyink and Hershberger, 
and give several applications of our results to combinatorial and algorithmic questions including to the \emph{multi-hitting set}
problem involving points and non-piercing regions.
\end{abstract}

\section{Introduction}
\label{sec:intro}
A fundamental algorithmic design technique in Computational Geometry introduced by Bentley and Ottmann~\cite{bentley1979algorithms}
is \emph{line sweep}.
In the original context in which the technique was invented, namely reporting all intersections among $n$
line segments in the plane, the basic idea was to move a vertical line from $x=-\infty$ to $x = +\infty$, stopping at {\em event points} (like end-points of segments or intersections of two segments) where updates are made.
Since then, the line sweep technique has found a wide variety of applications in Computational Geometry including 
polygon triangulation, computing Voronoi diagrams and its dual Delaunay triangulations, and computing the volume of a union of regions to name a few. The method can also be generalized to higher dimensions. See the classic books on Computational Geometry~\cite{de2000computational,edelsbrunner1987algorithms,preparata2012computational} for more
applications. Besides applications in Computational Geometry, it is also a useful tool in combinatorial proofs. 
For a simple example, a one dimensional sweep can be used to show that the clique cover number and maximum independent set size in interval graphs are equal. 
In most applications, the sweep line moves continuously and covers the plane so that each point in the plane is on the sweep line exactly once.
There are also applications where instead of sweeping with a line, we sweep with a closed curve like a circle. For instance we can start with an infinite radius circle and shrink it continuously to its center. 

Edelsbrunner and Guibas~\cite{edelsbrunner1986topologically} introduced the \emph{topological sweep} technique and showed
the advantage of sweeping with a \emph{topological line}, i.e., curve that at any point intersects each segment at most once rather
than a rigid straight line. They obtained better algorithms to report all intersection points between 
lines, which by duality improved results known at that time for many problems on point configurations.
Chazelle and Edelsbrunner~\cite{chazelle1992optimal} were also able to adapt the topological sweep ``in 20 easy pieces'' to report the $k$ intersections of $n$ given lines segments in $O(n\log n +k)$ time.

Thus, we can generalize the sweeping technique to sweep an arrangement of \emph{pseudolines}\footnote{A set $\mathcal{L}$ of bi-infinite curves in the plane is a collection of pseudolines if the curves in $\mathcal{L}$ pairwise intersect at most once.} with a pseudoline. 
We can generalize further, and do topological sweep with a closed curve.
In any arrangement of circles in the plane, any pair of circles intersect at two points, or don't intersect. Keeping this intersection property, while replacing the circles with simple Jordan curves, we obtain  an \emph{arrangement of pseudocircles} in which any pair of
curves either intersect in two points, or don't intersect. 
Note that for any two disks defined by circles or pseudocircles, $D_1$, $D_2$, the difference $D_1\,\setminus\, D_2$, is a path-connected set.

Showing that it is possible to sweep an arrangement of pseudolines with one of the pseudolines, and an arrangement of pseudocircles with one of the pseudocircles is much more challenging than sweeping with a line or a circle. Snoeyink and Hershberger~\cite{SnoeyinkH90}
showed in a celebrated paper that we can sweep an arrangement of two-intersecting curves with any curve in the arrangement. 
The result of Snoeyink and Hershberger has found several applications both in Computational Geometry as well as 
in combinatorial proofs. See for example,
\cite{ackerman2021coloring,agarwal2004lenses,buzaglo2008s,chiu2023coloring,felsner1999triangles,scheucher2020points,keller2026} and references therein. Some more applications appear in Section~\ref{sec:applications} of this paper.

It is natural to attempt to generalize the result of Snoeyink and Hershberger to $k$-intersecting curves for $k>2$, i.e., sweep an arrangement of $k$-intersecting
curves with a sweep curve that maintains the invariant that the arrangement is $k$-intersecting throughout the sweep. Unfortunately,
Snoeyink and Hershberger~\cite{SnoeyinkH90} show that this is not possible for $k\ge 3$, and ask
at the end of their paper what intersection property the sweep satisfies if we sweep an arrangement of $k$-intersecting curves. 
While we do not fully answer their question, we extend their result to show that we can maintain a more general property than 2-intersection while
sweeping. Specifically, we show that if the input curves satisfy a topological condition of being 
\emph{non-piercing}, then we can sweep with a curve so that the arrangement satisfies this property
throughout the sweep.
For a Jordan curve $\gamma$, let $\tilde{\gamma}$ denote the bounded region defined
by $\gamma$. Two Jordan curves $\alpha,\beta$ are said to be \emph{non-piercing} if $\tilde{\alpha}\,\setminus\,\tilde{\beta}$ is
a path connected region, and a set of curves is non-piercing if every pair of curves in the arrangement is non-piercing. The main theorem we prove in this paper is the following. The notion of `sweeping' referred to in the theorem is formally defined in Section~\ref{sec:prelim}. 

\begin{restatable}[Sweeping non-piercing arrangements]{theorem}{np}
\label{thm:mainthm}
Let $\Gamma$ be a finite set of non-piercing curves. 
Given any $\gamma \in \Gamma$ and a point $P \in \tilde\gamma$, we can sweep $\AG$ with $\gamma$ so that throughout the process, $P$ remains within $\tilde{\gamma}$ and 
the curves remain non-piercing  except at a finite number of instants in time.
\end{restatable}

While Theorem~\ref{thm:mainthm} talks about sweeping $\gamma$ inwards into $\tilde{\gamma}$, we can also sweep $\gamma$ in its exterior
by first applying an \emph{inversion}\footnote{\url{https://en.wikipedia.org/wiki/Inversive_geometry}} of the plane, and then applying Theorem~\ref{thm:mainthm} (See the Remark following the proof of Theorem~\ref{thm:mainthm} for a different approach).

We also give an alternate proof of the fundamental result of Snoeyink and Hershberger. The new proof introduces two conceptual tools, namely \emph{minimal lens untangling} and \emph{minimal triangle untangling} that may help in other settings involving pseudodisks or non-piercing regions. We also believe that the proof is conceptually simpler than that of 
Snoeyink and Hershberger, albeit still not very short. 

For many algorithmic applications, especially for several packing and covering problems, the restriction that
an arrangement is non-piercing is not harder than the corresponding problems for pseudodisks. For example, hypergraphs defined
by points and non-piercing regions enjoy a linear \emph{shallow-cell complexity}~\cite{DBLP:journals/dcg/RamanR20}
and therefore, admit $\epsilon$-nets of linear size and $O(1)$-approximation algorithms for covering problems via \emph{quasi-uniform}
sampling~\cite{DBLP:conf/stoc/Varadarajan10,DBLP:conf/soda/ChanGKS12}.
In particular, we believe that non-piercing is a more elementary condition than two-intersecting 
to which known results for pseudodisks can be extended. For example Har-Peled~\cite{DBLP:journals/siamcomp/Har-PeledQ17} and Chan and Grant~\cite{DBLP:journals/comgeo/ChanG14} showed that several covering problems involving simple
geometric regions where the regions are \emph{piercing} (i.e., not non-piercing), are APX-hard. 

The paper is organized as follows:
The notation used in the paper is described in Section~\ref{sec:prelim}. 
Section~\ref{sec:sweeping} describes the sweeping operations and Section~\ref{sec:basicops} describes the basic operations
required in the proof. Section~\ref{sec:nonpiercing} contains the proof of Theorem~\ref{thm:mainthm}, and Section~\ref{sec:snoeyinkthm}
contains a simpler proof of the Theorem of Snoeyink and Hershberger on two-intersecting curves. 
We describe applications of our results in Section~\ref{sec:applications}, and 
we conclude in Section~\ref{sec:conclusion} with some open questions and future directions.

\section{Preliminaries}
\label{sec:prelim}

Let $\mathbb{S}^1$ denote a circle centered at the origin.
A Jordan curve $\gamma$ is a continuous injective map from $\mathbb{S}^1\to\mathbb{R}^2$. 
By the Jordan curve theorem (see~\cite{DBLP:books/daglib/0030489}, Chapter 2), $\gamma$ splits the plane into two parts - a bounded
part (called the \emph{interior} of $\gamma$), denoted $\tilde\gamma$~\footnote{A region $r\subseteq\mathbb{R}^2$ is said to be \emph{simply-connected} if any closed loop can be continuously deformed
to a point inside the region.}), and an unbounded part
called the \emph{exterior} of $\gamma$. We orient $\gamma$ so 
that the bounded region lies to its left, and
we assume that the map
from $\mathbb{S}^1$ to $\gamma$ is such that the origin is mapped to the interior of $\gamma$.

Unless otherwise stated, by a \emph{curve}, we will mean an oriented Jordan curve. 
A \emph{traversal} of a curve is a walk starting at an arbitrary point on the curve,
walking in the direction of its orientation 
and returning to the starting point.

In the following, we denote by $\Gamma$, a finite set of oriented Jordan curves, and we use
$\tilde\Gamma$ to denote the set $\{ \tilde\gamma: \gamma \in \Gamma\}$ of bounded regions defined
by the curves in $\Gamma$. 
We assume throughout this paper that the curves in $\Gamma$ are in \emph{general position}, i.e., no three curves intersect at a point,
any pair of curves intersect in a finite number of points, and they intersect transversely (i.e., they cross) at these points.

\smallskip\noindent{\bf Arrangement.}
The \emph{arrangement} $\AG$ defined by $\Gamma$ is the decomposition of $\mathbb{R}^2$
into relatively open \emph{cells} of dimensions $0,1$ and $2$ induced by $\Gamma$, where each cell is
a maximal connected set of points lying in the intersection of a fixed subset of $\tilde\Gamma$. 
For convenience, we use the phrase `arrangement of $\Gamma$' to refer to $\AG$. The usual term
for 
cells of dimensions $0, 1$ and $2$ are respectively, the vertices, edges, and faces of 
 $\AG$. In this paper we use the usual terminology of \emph{vertices} and \emph{edges} to
 refer to the cells of dimension $0$ and $1$ respectively, but
 we use the term \emph{cell} to denote the faces (cells of dimension 2) in $\AG$. 
 The curves which contribute one or more edges to the boundary of a cell $C$ are said to \emph{define} the cell. If a curve $\alpha$ contributes an edge to the
boundary of a cell $C$, we say that $C$ \emph{lies on} $\alpha$. 

\smallskip\noindent
{\bf Digons, Lenses, and Triangles.}
A {\em digon} is a {\em bounded} cell $D$ in the arrangement of any two curves $\alpha, \beta \in \Gamma$ so that both $\alpha$ and $\beta$ contribute exactly one arc to the boundary of $D$. 
Note that $D$ is not necessarily a cell in $\AG$. 
If a cell in $\AG$ 
is a digon, we call it a {\em digon cell}. 
A digon defined by curves $\alpha$ and $\beta$ is called a {\em lens} if it is contained in both $\tilde{\alpha}$ and $\tilde{\beta}$ and it is called a {\em negative lens} if it is not contained in either $\tilde{\alpha}$ or $\tilde{\beta}$. There is a third possibility, 
namely that a digon is contained in exactly one of $\tilde\alpha$ or $\tilde\beta$, in which case it is called a \emph{lune}.
A {\em triangle} is a {\em bounded} cell $T$ in the arrangement of three curves in $\Gamma$ so that each curve contributes exactly one arc to the boundary of $T$. Note that a triangle may or may not belong to the region bounded by any of curves defining it.
A triangle need not be a cell in $\AG$. If a cell in $\AG$ is a triangle, we call it a {\em triangle cell}. 

\smallskip\noindent
{\bf Sweeping.}
Given an arrangement $\AG$ of oriented Jordan curves, a sweeping curve $\gamma\in\Gamma$, a property $\Pi$ 
satisfied by the curves in $\Gamma$, and a point $P\in\tilde{\gamma}$, a `sweep of $\AG$ by $\gamma$' is a continuous
movement of $\gamma$, keeping the other curves in $\Gamma$ fixed, until the region bounded by $\gamma$ shrinks to the point $P$ so that
\begin{enumerate}[(i)]
\item for any two times $t>t'$, the curve $\gamma$ at time $t$ is contained in the region bounded by $\gamma$ at time $t'$, and 
\item except at a finite set of instances in time,  $\AG$ satisfies the property $\Pi$.
\end{enumerate}

In this paper, we assume that the initial arrangement of curves is in general position and the property $\Pi$ is that the curves are `non-piercing and in general position'. 

\medskip\noindent
{\bf Pseudocircles and Non-piercing regions.}
If the curves in $\Gamma$ intersect pairwise in at most two points, we call 
$\Gamma$ a set of \emph{pseudocircles}, and $\tilde\Gamma$ is called a set of \emph{pseudodisks}. 
Let $\alpha$ and $\beta$ be two curves. 
We say that $\alpha$ and $\beta$ are \emph{non-piercing}
if both $\tilde{\alpha}\,\setminus\,\tilde{\beta}$, and $\tilde{\beta}\,\setminus\,\tilde{\alpha}$ are path-connected. Otherwise, they are said to be \emph{piercing}. 
A set of curves is said to be \emph{non-piercing} if the curves in the set are pairwise non-piercing. 
In this paper, we restrict to a finite set of non-piercing curves in general position.
Figure~\ref{fig:nonpiercing} shows an example of a pair of pseudocircles, a 
pair of non-piercing curves, and a pair of piercing curves.
\begin{figure}[h!]
\begin{center}
\begin{subfigure}[t]{0.31\textwidth}
\begin{center}
\includegraphics[width=0.95\textwidth]{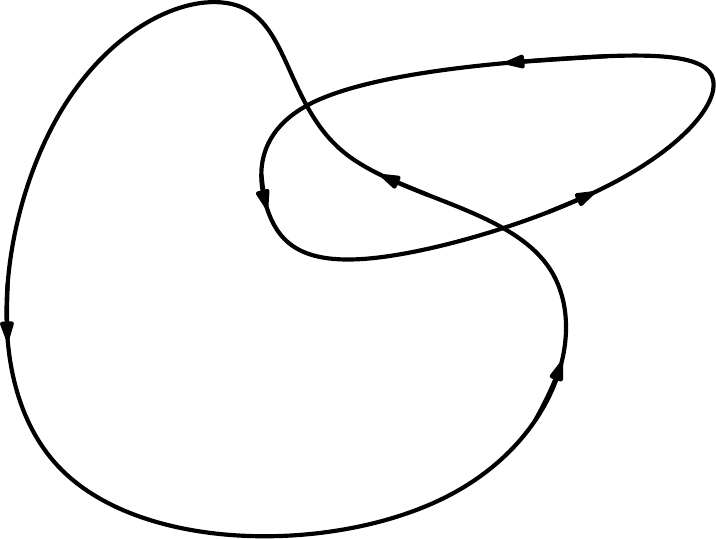}
\subcaption{Pseudocircles}
\end{center}
\end{subfigure}%
\hspace{0.02\textwidth}
\begin{subfigure}[t]{0.3\textwidth}
\begin{center}
\includegraphics[width=0.95\textwidth]{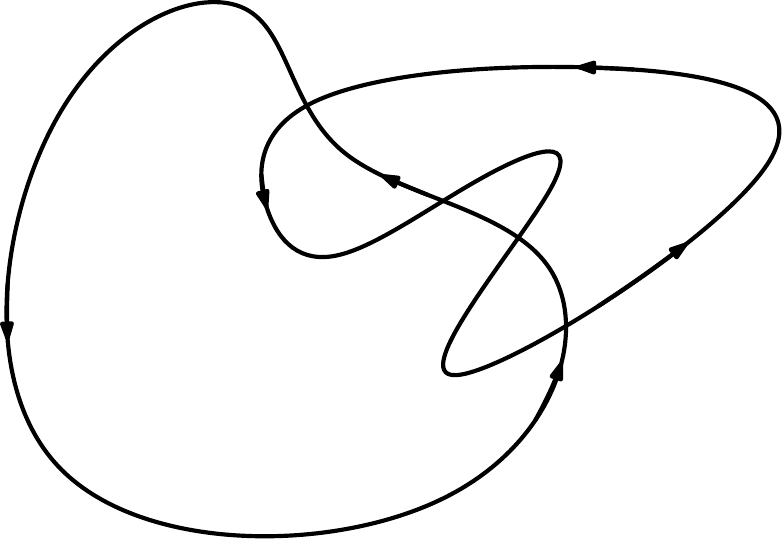}
\subcaption{Non-piercing curves}
\end{center}
\end{subfigure}%
\hspace{0.02\textwidth}
\begin{subfigure}[t]{0.31\textwidth}
\begin{center}
\includegraphics[width=0.95\textwidth]{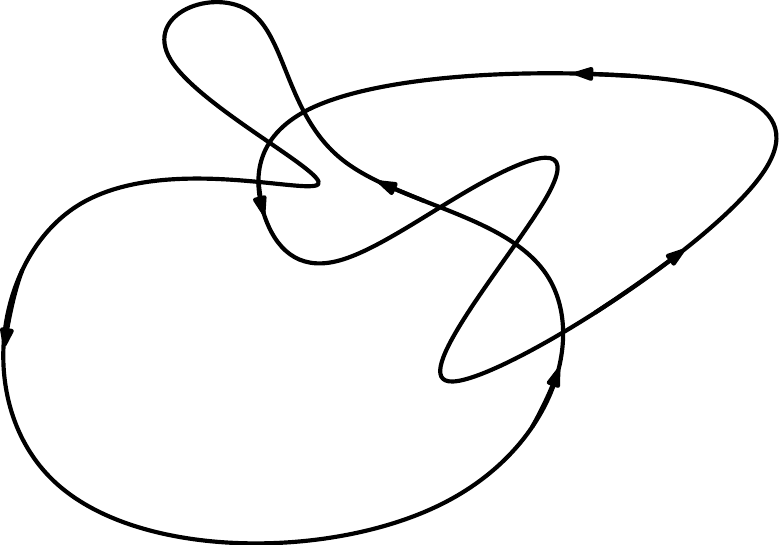}
\subcaption{Piercing curves}
\end{center}
\end{subfigure}
\caption{The figure above shows an example of a pair of pseudocircles, a pair of non-piercing curves, and a pair of piercing curves.}
\label{fig:nonpiercing}
\end{center}
\end{figure}

\smallskip\noindent
{\bf Reverse-cyclic sequences.}
Let $\alpha$ and $\beta$ be two curves which intersect at a finite number of points $x_1, \cdots, x_k$ transversally. 
Let $\sigma_\alpha(\beta)$ be the cyclic sequence of the intersection points of $\alpha$ and $\beta$ encountered in a traversal of $\alpha$. 
We define $\sigma_\beta(\alpha)$ analogously. We say that  `the intersections of $\alpha$ and $\beta$ are  
\emph{reverse-cyclic}' if $\sigma_{\alpha}(\beta)$ and $\sigma_{\beta}(\alpha)$ are the reverse of each other up to a cyclic shift. The following lemma was proved by Basu-Roy, et al.~\cite{DBLP:journals/dcg/RoyGRR18}. See Figure~\ref{fig:enter-label}.
\begin{lemma}[\cite{DBLP:journals/dcg/RoyGRR18}]
\label{lem:revcyclic}
Two curves $\alpha$ and $\beta$ that intersect at a finite number of points transversally are non-piercing iff their intersections are reverse-cyclic.
\end{lemma}
\begin{figure}[!h]
    \centering
    \includegraphics[scale=0.5]{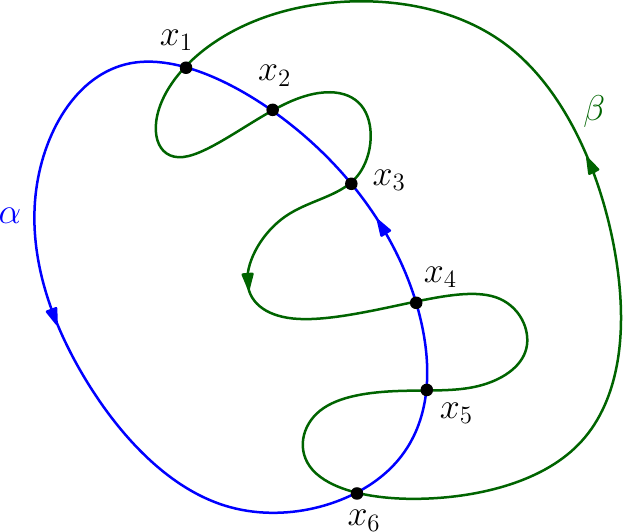}
    \caption{\centering Illustration of Lemma~\ref{lem:revcyclic}. Here $\sigma_\alpha(\beta)=(x_6,x_5,x_4,x_3,x_2,x_1)$ and $\sigma_\beta(\alpha)=(x_1,x_2,x_3,x_4,x_5,x_6)$ are reverse-cyclic.}
    \label{fig:enter-label}
\end{figure}

\section{Sweeping}
\label{sec:sweeping}
Let $\Gamma$ be a set of non-piercing curves in the plane in general position, let $\gamma \in \Gamma$ be one of the curves called the {\em sweep curve} and let $P$ be a specified point in $\tilde\gamma$.
Recall that by `sweeping of $\AG$ by $\gamma$',
we mean the process of continuously shrinking $\tilde{\gamma}$ to $P$ while keeping the other curves fixed and maintaining the non-piercing property throughout except a finite set of instants in time.

Snoeyink and Hershberger~\cite{SnoeyinkH90} discretized this continuous sweeping process into a set of discrete operations each of which can be implemented as a continuous deformation of $\gamma$ over a unit time interval. 
We follow their arguments and likewise describe a set of allowable discrete {\em sweeping operations}. Most of our operations are slight variations of their
operations, except one new operation that we introduce. 

In order to describe the operations, we need one more definition. 
We say that point $v$ in $\tilde{\gamma}$ is \emph{visible} from the sweep curve $\gamma$
if there is an arc $\tau(u,v)$ (called the {\em visibility arc}) joining a point $u$ on $\gamma$ with $v$
such that the interior of $\tau(u,v)$ lies in $\tilde{\gamma}$
and does not intersect any of the curves in $\Gamma$.
We say that a curve $\alpha$ is visible from $\gamma$ if some point on $\alpha$ is visible from $\gamma$.

We define the following discrete sweeping operations.

\begin{enumerate}[(i)]
\item \emph{Take a new loop:} Let $\alpha$ be a curve that does not intersect $\gamma$ and such that $\alpha$ is visible from $\gamma$ 
via the visibility arc $\tau(u,v)$. We define ``taking a new loop'' as the operation that modifies
$\gamma$ to $\gamma'$ as follows: taking two points $u',u''$ arbitrarily close to $u$ on either side to $u$ on $\gamma$, 
we replace the segment of $\gamma$ between $u'$ and $u''$ containing $u$ by a curve between $u'$ and
$u''$ that lies arbitrarily close to $\tau(u,v)$ and loops around $v$ crossing $\alpha$ twice. 
See Figure~\ref{fig:takeLoop}. 
More formally, $\gamma$ is modified to $\gamma'$ 
so that $\tilde{\gamma'} = \tilde{\gamma} \,\setminus\, (\tau(u,v) \oplus B_\epsilon)$ where $\oplus$ denotes Minkowski sum and $B_\epsilon$ is a ball of radius $\epsilon$ for an arbitrarily small $\epsilon$. In the rest of the paper, we avoid such formal definitions and resort to figures for ease of exposition. However, our informal definitions using figures can easily be formalized.

\begin{figure}[ht!]
\begin{subfigure}{0.4\textwidth}
\centering
 \includegraphics[scale=.5]{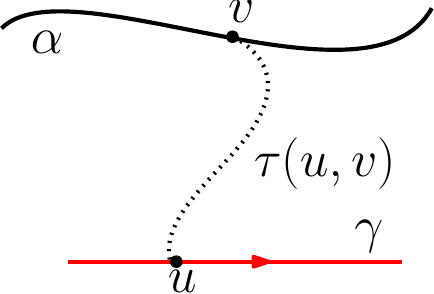}
\label{fig:takeLoop1}
\end{subfigure}\hspace{1cm}
\begin{subfigure}{0.5\textwidth}
\centering
\includegraphics[scale=.5]{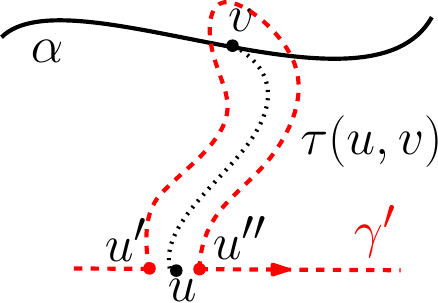}
\label{fig:takeLoop2}
 \end{subfigure}
 \caption{\centering Taking a new loop.}
 \label{fig:takeLoop}
 \end{figure}

\item \emph{Bypass a digon cell:} 
Suppose that a curve $\alpha \in \Gamma$ and the sweep curve $\gamma$ form a digon cell $D$ in $\tilde\gamma$ with vertices $u$ and $v$. Then, we define the operation of 
``bypassing $D$'' as the modification of $\gamma$ to a curve $\gamma'$ that goes around the digon cell $D$ as shown in Figure~\ref{fig:emptyDigon}.

\begin{figure}[!ht]
\begin{subfigure}{0.4\textwidth}
\centering
 \includegraphics[scale=.5]{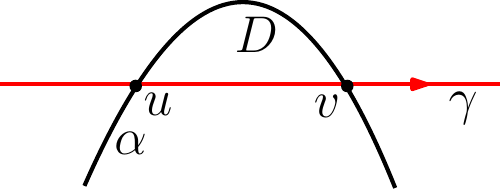}
\label{fig:emptyDigon1}
\end{subfigure}\hspace{1cm}
\begin{subfigure}{0.5\textwidth}
\centering
\includegraphics[scale=.5]{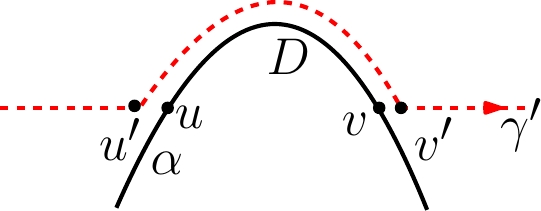}
\label{fig:emptyDigon2}
 \end{subfigure}
 \caption{\centering Bypassing a digon cell. Note that the orientation of $\alpha$ doesn't matter.}
 \label{fig:emptyDigon}
 \end{figure}

\item \emph{Bypass a triangle cell:} Let $\alpha$ and $\beta$ be two curves that along with the sweep curve $\gamma$ define a triangle cell $T$ in
$\tilde\gamma$. We define the operation of ``bypassing $T$'' as the modification of $\gamma$ to a curve $\gamma'$ which goes around $T$ as show in Figure~\ref{fig:emptyTriangle}.
    
\begin{figure}[!ht]
\begin{subfigure}{0.4\textwidth}
\centering
 \includegraphics[scale=.45]{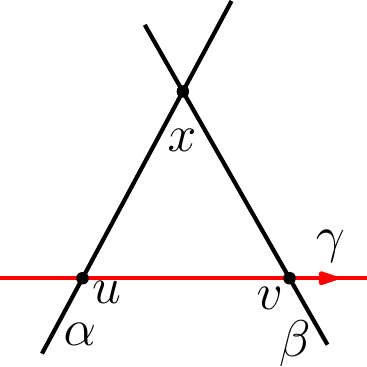}
\label{fig:emptyTriangle1}
\end{subfigure}\hspace{1cm}
\begin{subfigure}{0.5\textwidth}
\centering
\includegraphics[scale=.45]{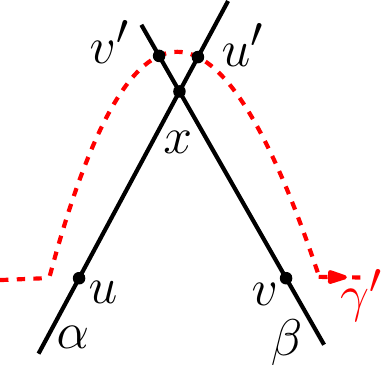}
\label{fig:emptyTriangle2}
 \end{subfigure}
 \caption{\centering Bypassing a triangle cell.}
 \label{fig:emptyTriangle}
 \end{figure}

 \item \emph{Bypass a visible vertex:} 
Let $v$ be a vertex of $\AG$ in $\tilde\gamma$ defined by the curves $\alpha, \beta \in \Gamma \,\setminus\, \{\gamma\}$.
Suppose there is a visibility arc $\tau(u,v)$ from a point $u$ on $\gamma$ to $v$ whose interior lies 
in the interior of both $\tilde{\alpha}$ and $\tilde{\beta}$. 
By definition of a visibility arc, the interior of $\tau(u,v)$ also lies in the interior of $\tilde{\gamma}$. 
In this case, we refer to $\tau(u,v)$ as the {\em bypassability curve} for $v$ and define ``bypassing $v$'' as the modification of $\gamma$ to a curve $\gamma'$ that takes a loop around $v$ as shown in Figure~\ref{fig:sweepVisible}.
Note that the orientation of the curves is important here.

\begin{figure}[!ht]
\begin{subfigure}[!htbp]{0.4\textwidth}
\centering
 \includegraphics[scale=.6]{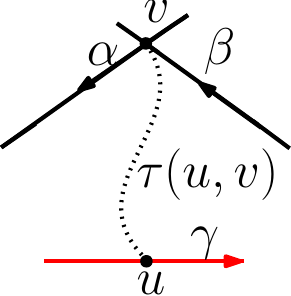}
\subcaption{\centering visibility curve $\tau(u,v)$}
\label{fig:sweepVisible1}
\end{subfigure}\hspace{1cm}
\begin{subfigure}[!htbp]{0.5\textwidth}
\centering
\includegraphics[scale=.6]{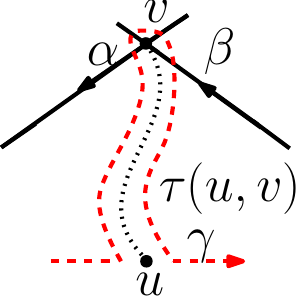}
\subcaption{\centering  $\gamma'$ after bypassing vertex $v$}
\label{fig:sweepVisible2}
 \end{subfigure}
 \caption{\centering Bypassing a visible vertex.}
  \label{fig:sweepVisible}
 \end{figure}
\end{enumerate}

\hide{
\section{Simplifying transformations}
\label{sec:transform}

Let $\Gamma$ be a set of non-piercing regions in the plane. In Section~\ref{sec:sweeping}, we described sweeping either with
a closed curve, or with a bi-infinite curve. In this section, we show that we can, without loss of generality,
restrict our attention to sweeping with a bi-infinite curve. The results in this section are essentially from the results of
Snoeyink and Hershberger and are described here for completeness.

\begin{lemma}
\label{lem:transform}
`Let $\Gamma$ be an arrangement of non-piercing curves with sweep curve $\gamma\in\Gamma$. Suppose $\gamma$ is a closed curve, then 
we can transform the instance into a collection of non-piercing curves $\Gamma'$ s.t. in $\Gamma'$ the image of $\gamma$ is a
bi-infinite curve.
\end{lemma}
\begin{proof}

\end{proof}

When sweeping with a bi-infinite curve $\gamma$, by applying a homeomorphism of the plane, we can assume that $\gamma$ is a horizontal line
and that we are sweeping upwards. In such a case, the intersection between the curves in $\Gamma$ below $\gamma$ don't play any role, and
we can remove them. In particular, we cut each curve below $\gamma$ and let it run to $-\infty$.

\begin{lemma}
\label{lem:simplify}
Given an arrangement $\Gamma$ of non-piercing regions, and a bi-infinite curve $\gamma$, we can modify the arrangement so that
the curves do not intersect outside $\tilde{\gamma}$.
\end{lemma}
}

The following lemma summarizes properties of the above operations that are intuitively obvious from their descriptions.

\begin{restatable}{lem}{ops}
\label{lem:ops}
Let $\Gamma$ be a set of non-piercing curves in general position and let $\gamma\in\Gamma$ be the sweep curve. Let $\gamma'$ be the modified sweep curve obtained after
applying one of the operations: $(i)$ taking a new loop, $(ii)$ bypass 
 a digon cell, $(iii)$ bypass a triangle cell, or $(iv)$ bypass a visible vertex. Then, the resulting set of curves $(\Gamma\,\setminus\,\{\gamma\})\cup\{\gamma'\}$ is non-piercing and in general position. Furthermore,
 operation $(i)$ increases by $2$ the number of intersection points between $\gamma$ and the curve $\alpha$ we take a new loop on, operation $(ii)$ decreases by $2$ the number of intersections between $\gamma$ and the other curve $\alpha$ defining the digon cell, operation $(iii)$ does not change the number of intersections between $\gamma$ and other curves, and operation $(iv)$ increases by $2$ the number of intersections between $\gamma$ and each of the curves $\alpha, \beta$ defining the vertex being bypassed.
\end{restatable}
\begin{proof}
Consider operation $(i)$ and let $\alpha$ be the curve on which we are taking a new loop. 
As shown in Figure~\ref{fig:takeLoop}, this increases the number of intersections between $\alpha$ and $\gamma$ by $2$ by inserting two 
intersection points along either curve. The two new intersections points appear consecutively along either curve but in opposite order along their orientations. This ensures that the intersections between $\alpha$ and $\gamma$ remain reverse-cyclic, which by Lemma~\ref{lem:revcyclic} implies that $\alpha$ and $\gamma$ remain non-piercing after the operation. Since no other curve is affected, the entire arrangement remains non-piercing.  

Now consider operation $(ii)$ and let $\alpha$ be the curve defining the digon cell along with the sweep curve $\gamma$. 
As Figure~\ref{fig:emptyDigon} shows,
in this case we remove two intersections points (namely the vertices of the digon cell) which appear consecutively along both curves. As a result the intersections of $\alpha$ and $\gamma$ remain reverse-cyclic and by Lemma~\ref{lem:revcyclic} $\alpha$ and $\gamma$ remain non-piercing. No other curve is affected.

Next, consider operation $(iii)$. In this case let $\alpha$ and $\beta$ be the curves defining the triangle along with the sweep curve $\gamma$. As shown in Figure~\ref{fig:emptyTriangle}, the intersection points $u$ and $v$ of $\gamma$ with $\alpha$ and $\beta$ effectively move to $u'$ and $v'$ switching their order along $\gamma$. However the sequence or the number of intersections of $\gamma$ with $\alpha$ or $\beta$ (or any other curve) does not change. Thus the arrangement remains non-piercing with the same number of intersection points as before.

Finally, consider operation $(iv)$. As shown in Figure~\ref{fig:sweepVisible}, the curve $\gamma$ is modified to go around the vertex $v$. This increases by $2$ the number of intersections between $\alpha$ and $\gamma$ and between $\beta$ and $\gamma$. If we consider the pair of curves $\alpha$ and $\gamma$ and focus on the intersections between them (i.e., we ignore intersections with other curves), the two new intersections are consecutive along both curves. Therefore, the intersections between $\alpha$ and $\gamma$ remain reverse-cyclic implying that $\alpha$ and $\gamma$ remain non-piercing. Analogously, $\beta$ and $\gamma$ remain non-piercing. 

In all the operations, the curve $\gamma$ is the only curve that is modified and modified curve does not pass through the intersection of two other curves or intersect any of the curves tangentially. The entire arrangement therefore, remains in general position.
\end{proof}

\smallskip\noindent
{\bf Remark.} Lemma~\ref{lem:ops} shows that the arrangement
remains `non-piercing and in general position' 
at the end of any of the four operations. 
As a continuous process however, this property is violated at discrete
points in time. It happens when the sweep curve passes through the intersection point of two other curves 
(while bypassing a triangle cell or a visible vertex). It also happens when the sweep curve shares a boundary point with another curve (while taking a new loop, or bypassing a digon cell). 

\begin{corollary}
\label{cor:two}
Let $\Gamma$ be a set of pseudocircles in general position with sweep curve $\gamma\in\Gamma$. Then, after applying any one of the three
operations: (i) take a new loop, (ii) bypass a digon cell, or (iii) bypass a triangle cell, the modified set of curves ($\Gamma\,\setminus\,(\{\gamma\}\cup\{\gamma'\})$), where $\gamma'$ is the modified sweep curve, is a set of pseudocircles in general position.
\end{corollary}

\begin{proof}
We need to show that after any of the operations, the number of intersections between any two curves is either $0$ or $2$. Since $\gamma$ is the only curve that is being modified, we only need to consider pairs of curves involving $\gamma'$ in the modified set of curves.
Operation $(i)$, i.e., taking a new loop is only applied on a curve $\alpha$ that does not intersect $\gamma$. Therefore, after the operation, $\gamma'$ and $\alpha$ intersect in two points.
Operation $(ii)$, i.e., bypass a digon cell, is applied to a digon cell 
formed by $\gamma$ and a curve $\alpha$. After the operation, $\gamma'$ and $\alpha$ do not intersect. 
Operation $(iii)$, i.e., bypass a triangle cell is applied to a triangle cell formed by curves $\alpha$ and $\beta$  with $\gamma$.
After the operation, the number of intersections between $\alpha$ and $\gamma'$, and between $\beta$ and $\gamma'$ remain unchanged.
Thus, in all cases, the modified set of curves is a set of pseudocircles. The figures corresponding to the operations show that the modified set of curves is in general position.
\end{proof}

\hide{
Figure~\ref{fig:vertex} shows that in the case of non-piercing curves, unlike the case with two-intersecting curves,
we sometimes require the operation of bypassing a vertex. 

\begin{figure}[ht!]
\begin{center}
    \includegraphics[width=2in]{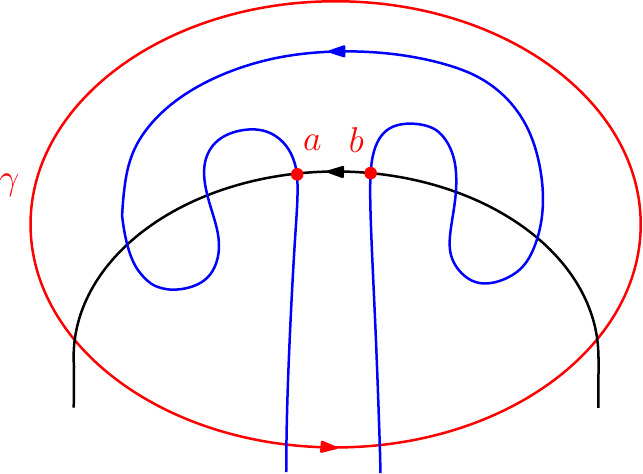}
    \caption{\centering An example of non-piercing curves where the only operation available for the sweep is to bypass
    either the vertex labeled $a$ or the vertex labeled $b$.}
    \label{fig:vertex}
\end{center}
\end{figure}
}
\section{Minimal Digons and Triangles}
\label{sec:basicops}
In this section, we describe two basic operation in an arrangement of a set $\Gamma$ of non-piercing curves.

\smallskip\noindent
{\bf Minimal digon untangling.}
A digon $L$ is said to be \emph{minimal} if it does not contain another digon.
Let $L$ be a minimal digon defined by curves $\alpha$ and $\beta$ in $\Gamma$ with vertices $u$ and $v$. By the minimality of $L$, the intersection of any curve $\delta \in \Gamma$ with $L$
consists of disjoint arcs, each having one end-point on $\alpha$ and one end-point on $\beta$.
The operation of untangling $L$ modifies
the arcs of $\alpha$ and $\beta$ as shown in 
Figure~\ref{fig:digonbypassing}. 
Formally, the untangling is done in two steps. First we replace the arc of $\alpha$ between $u$ and $v$ by the arc of $\beta$ between $u$ and $v$ and vice versa so that $u$ and $v$ are now points of tangencies between $\alpha$ and $\beta$. Next we get rid of the tangencies by moving $\alpha$ and $\beta$ slightly apart from each other
around $u$ and $v$. 
We also add {\em reference points} called  $u_\alpha$ and $u_{\beta}$ arbitrarily close to $u$ so that $\alpha$ passes between $u$ and $u_\alpha$ keeping $u_\beta$ on the same side as $u$, and similarly $\beta$ passes between $u$ and $u_\beta$ keeping $u_\alpha$ on the same side as $u$. We define the reference points $v_\alpha$ and $v_\beta$ close to $v$ analogously. Note that the reference points are not vertices in the modified arrangement. 
We claim that if we are given an arrangement of non-piercing curves, then the arrangement obtained by
untangling a minimal digon remains non-piercing. Before we do that, we need the following Lemma from Basu Roy et al.,~\cite{DBLP:journals/dcg/RoyGRR18}.

\begin{figure}[ht!]
\begin{subfigure}{0.3\textwidth}
\centering
 \includegraphics[scale=.3]{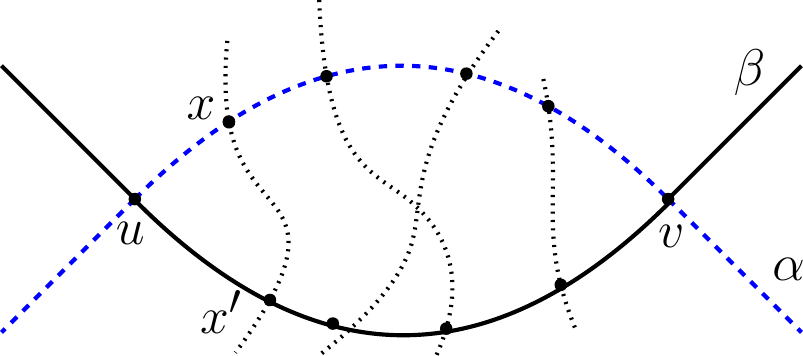}
\end{subfigure}\hspace{0.2cm}
\begin{subfigure}{0.3\textwidth}
\centering
\includegraphics[scale=0.3]{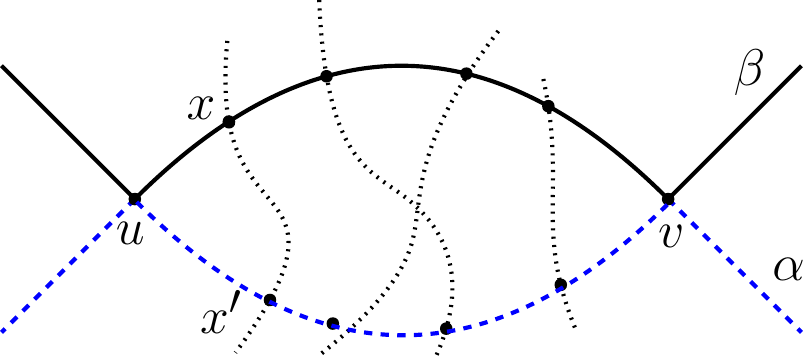}
 \end{subfigure}\hspace{0.2cm}
 \begin{subfigure}{0.3\textwidth}
\centering
\includegraphics[scale=0.3]{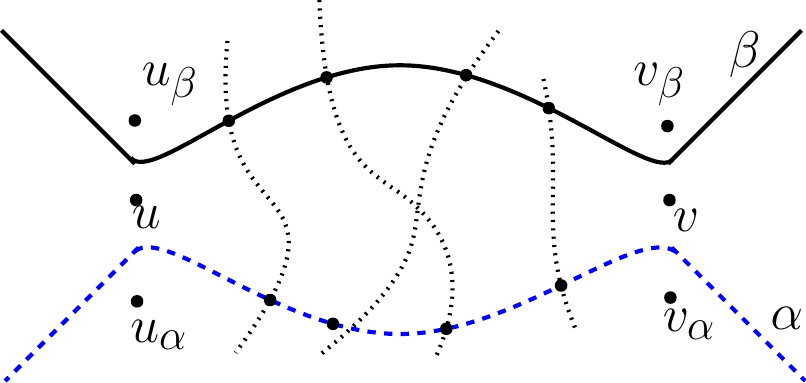}
 \end{subfigure}\hspace{0.2cm}
 \caption{The operation of untangling a digon defined by curves $\alpha$ and $\beta$.
 In the middle figure, the curves $\alpha$ and $\beta$ are swapped around the lens. In the third
 figure, the curves are separated slightly away from the vertices $u$ and $v$.
Note that the orientation of $\alpha$ and $\beta$ are not important, and hence not shown.}
 \label{fig:digonbypassing}
 \end{figure}

\hide{
\begin{figure}[ht!]
\begin{subfigure}{0.4\textwidth}
\centering
 \includegraphics[scale=.4]{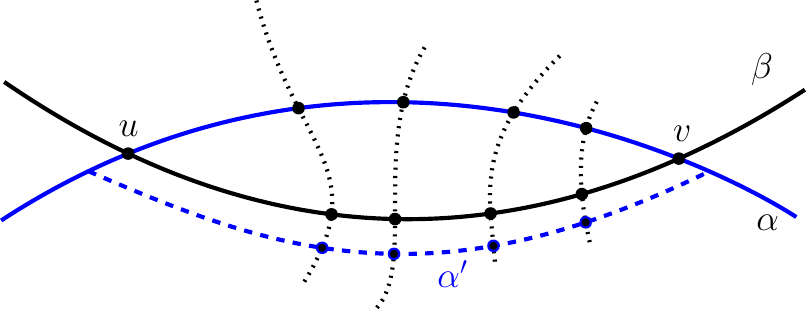}
\subcaption{The figure above shows the operation of bypassing a minimal lens $L_{uv}$ by $\alpha$.}
\label{fig:digonbypassing}
\end{subfigure}\hspace{1.2cm}
\begin{subfigure}{0.4\textwidth}
\centering
\includegraphics[scale=0.4]{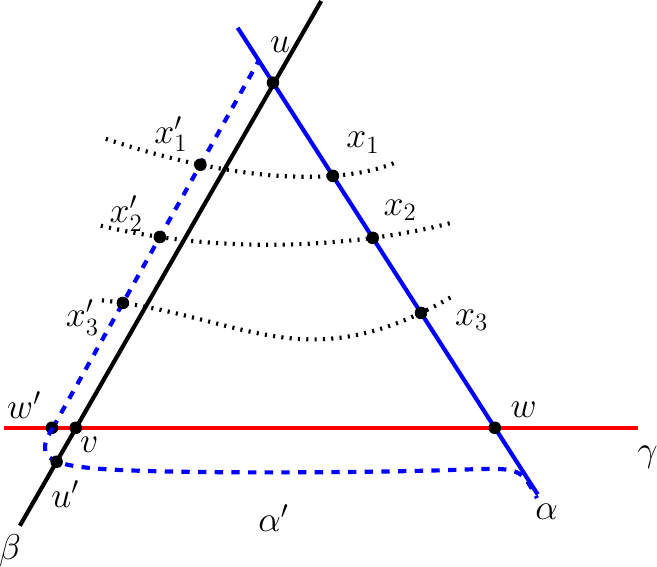}
    \subcaption{The operation of bypassing a minimal triangle by $\alpha$}
    \label{fig:mintrianglebypass}
 \end{subfigure}
 \subcaption{The two bypassing operations used in the proof. The curve $\alpha$ is replaced by the curve $\alpha'$, where the arc of $\alpha'$ that differs from $\alpha$ is shown by
 a dashed curve.}
 \end{figure}

\begin{figure}[!ht]
\begin{center}
\includegraphics[width=3in]{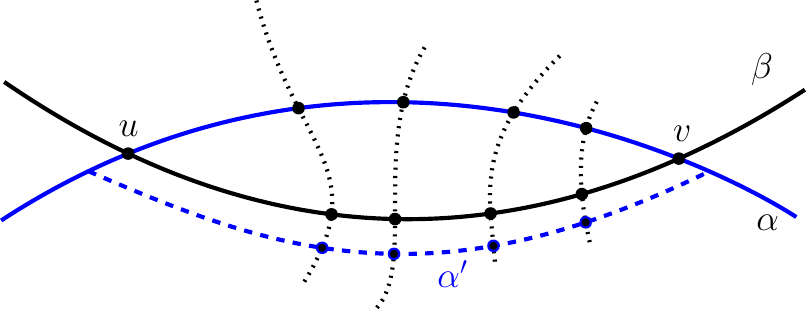}
\caption{}
\label{fig:digonbypassing}
\end{center}
\end{figure}
}

\begin{restatable}{lem}{minlenses}
\label{lem:digonbypass}
Let $L$ be a minimal digon formed by $\alpha$ and $\beta$ in the arrangement $\AG$ of $\Gamma$ as above. The modified arrangement of curves obtained by untangling $L$ is non-piercing. Furthermore, the number of intersections between $\alpha$ and $\beta$ decreases by $2$ and for any other pair of curves the number of intersections does not change.
\end{restatable}
\begin{proof}
Since $L$ is a minimal digon, the intersection of any other curve $\delta \in \Gamma \,\setminus\, \{\alpha, \beta\}$ with $L$ is a collection of disjoint segments with one vertex on $\alpha$ and the other on $\beta$.
Let $s$ be one of these segments with end point $x$ on $\alpha$ and $x'$ on $\beta$. Then, as a result of untangling, 
the earlier intersection of $\alpha$ and $\delta$ at $x$ is replaced by the intersection of the modified 
$\alpha$ and $\delta$ at $x'$. 
However, since $\alpha$ and $\delta$ do not have any intersections in the interior of the segment $s$, 
the order of the intersections of $\alpha$ and $\delta$ along either curve remains unchanged i.e., they remain reverse-cyclic. 
Hence, by Lemma~\ref{lem:revcyclic}, the modified $\alpha$ and $\delta$ are non-piercing. 
Analogously, the modified $\beta$ and $\delta$ are non-piercing. 
This also shows that the number of intersections between $\alpha$ and $\delta$ does not change. 
Similarly, the number of intersections between $\beta$ and $\delta$ does not change. 
The curves $\alpha$ and $\beta$ lose two consecutive points of intersections (namely the vertices of $L$) 
and therefore the order of the remaining intersections along them remains reverse-cyclic. 
Thus, the modified $\alpha$ and the modified $\beta$ are also non-piercing.
\end{proof}

The only new/changed cells in the modified arrangement obtained after untangling $L$ are those that contain one of the points $u, v, u_\alpha, u_\beta, v_\alpha$ or $v_\beta$. 
By `newly created' we mean the cell in the modified arrangement which are not there in the original arrangement. Thus, for example, when we say 'newly created digon cells or triangle cells' after an operation, we mean digon cells or triangle cells in the modified arrangement which were not already present in the  arrangement before the operation was executed. 
Note that there may also be cells in the modified arrangement that are geometrically identical to a cell in the original arrangement but their defining curves have changed  - $\alpha$ is replaced by $\beta$, or vice-versa. These are also considered `newly created'.

 \begin{restatable}{lem}{lenstriangle}
\label{lem:lenstriangle}
Let $\AG$ be an arrangement of a set of non-piercing curves $\Gamma$ with sweep curve $\gamma\in\Gamma$.
Let $L$ be a minimal digon in the arrangement formed by curves 
$\alpha,\beta\in\Gamma$ which have only two intersections 
 $u$ and $v$ in $\tilde\gamma$ (See Figure~\ref{fig:digonbypassing}).
In the arrangement obtained after untangling $L$, 
any newly created digon cells or triangle cells on the sweep curve $\gamma$ 
must contain one of the reference points $u_\alpha$, $u_\beta$, $v_\alpha$ or $v_\beta$. 
\end{restatable}
\begin{proof}
Let $\Gamma'$ denote the modified set of curves obtained on untangling $L$. 
The only cells in $\mathcal{A}(\Gamma')$ that are newly created (i.e., not already present in $\AG$) are those that have a side contributed by either $\alpha$ or $\beta$. 
Among those, any cell $C$ that does not contain the reference points $u_\alpha, u_\beta, v_\alpha, v_\beta$ or the vertices $u, v$, 
corresponds to a geometrically identical cell $C'$ in $\AG$ so that some of the sides contributed by $\alpha$ is $C'$ are contributed by the modified $\beta$ in $C$ or vice versa. Now, consider those cells that contain either $u$ or $v$. Note that any cell has sides contributed by
both of the modified curves $\alpha$ and $\beta$. Since by assumption these modified curves do not intersect in $\tilde\gamma$,
such a cell must have at least two more sides and hence is not a digon or a triangle cell. The lemma follows.
\hide{
Consider a cell $C$ in $\AG$ such that neither $\alpha$ nor $\beta$ contribute a side to $C$. Then,
the cell $C$ remains unchanged in $\mathcal{A}(\Gamma')$. Now, consi
If either $\alpha$ or $\beta$ contribute to a cell $C$. 
The cells in the new arrangement containing either $u$ or $v$ has sides contributed by $\alpha, \beta$, $\gamma$ and at least
one other curve (also possibly $\gamma$) and it thus has at least 4 sides.

any cell $C$ that does not contain the reference points $u_\alpha, u_\beta, v_\alpha, v_\beta$ or the vertices $u, v$, is identical to a cell in $\AG$ and hence the number of sides of such a cell remains unchanged. 
In $\mathcal{A}(\Gamma')$
one of their sides defined by $\alpha$ may have been replaced by $\beta$ or vice versa but this does not change whether the cell lies on $\gamma$. Thus such cells cannot be digon cells or triangle cells on $\gamma$ unless they were already so in the original arrangement (before bypassing $L$). 
The cells in the new arrangement containing either $u$ or $v$ has sides contributed by $\alpha, \beta$, $\gamma$ and at least
one other curve (also possibly $\gamma$) and it thus has at least 4 sides. Therefore the only cells that can become digon cells or triangle cells
on $\gamma$ are the ones with reference points $u_\alpha, u_\beta, v_\alpha$, or $v_\beta$.
}
\end{proof}

\smallskip\noindent{\bf Minimal triangle untangling.}
Let $T$ be a triangle in the interior $\tilde{\gamma}$ of the sweep curve $\gamma$, where $T$ is 
defined by the curves $\alpha, \beta, \zeta \in \Gamma$.
We say that $T$ is a ``minimal triangle on $\zeta$'', or ``base $\zeta$'' if 
the intersection of any other curve $\delta\in\Gamma\,\setminus\,\{\alpha,\beta,\zeta\}$ with $T$ 
consists of a set of disjoint arcs each with one end-point on $\alpha$ and one end-point on $\beta$. 
The operation of untangling $T$ modifies $\alpha$ and $\beta$ so that they go around $T$ and their intersection $u$ on the boundary of $T$ moves 
to the other side of $\zeta$, as shown in Figure~\ref{fig:mintrianglebypass}.
We show next, that untangling $T$ leaves the arrangement non-piercing.

\begin{figure}[h!]
\centering
\begin{subfigure}[!htbp]{0.45\textwidth}
\centering
 \includegraphics[scale=.4]{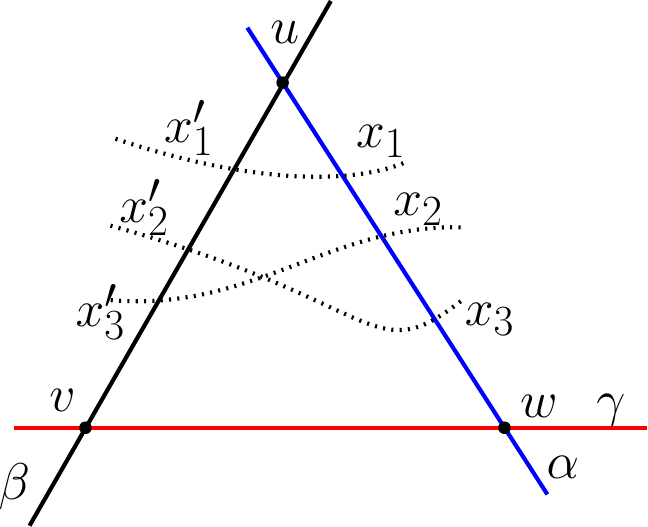}
\end{subfigure}\hspace{0.3cm}
\begin{subfigure}[!htbp]{0.45\textwidth}
\centering
\includegraphics[scale=0.4]{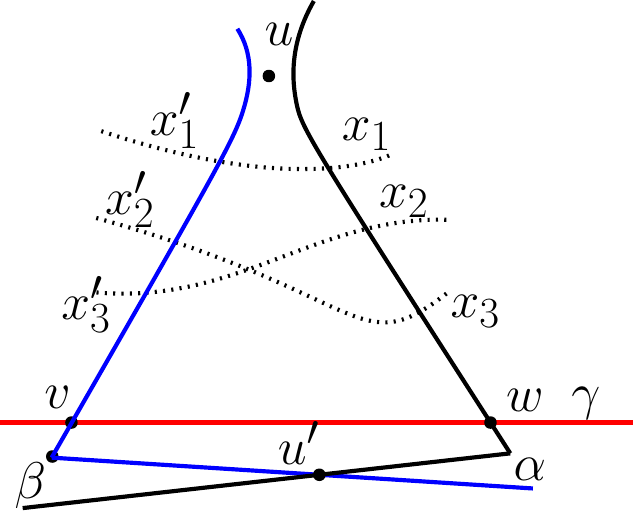}
 \end{subfigure}\hspace{0.3cm}
 \caption{\centering The operation of untangling a minimal triangle formed by curves $\alpha$ and $\beta$ on $\zeta$. 
 }
 \label{fig:mintrianglebypass}
 \end{figure}
\begin{restatable}{lem}{trianglebypass}
\label{lem:trianglebypass}
Let $T$ be a minimal triangle in an arrangement $\AG$ of non-piercing curves $\Gamma$ defined by curves $\alpha, \beta$ and
$\zeta$. Then, untangling $T$ yields a non-piercing arrangement. 
The number of intersection points between any pair of curves does not change. Furthermore, if $\zeta=\gamma$, and
$T$ was contained in $\tilde{\gamma}$, then
the number of intersection points in $\tilde{\gamma}$ decreases by 1. 
\end{restatable}
\begin{proof}
Let $u$ be the intersection point of $\alpha$ and $\beta$ on the boundary of $T$. 
Similarly, let $v$ and $w$ be the intersection points of $\beta$ and $\zeta$, and 
$\alpha$ and $\zeta$, respectively, on the boundary of $T$. 

By the definition of untangling $T$ (shown in Figure~\ref{fig:mintrianglebypass}), 
the intersection $u$ between $\alpha$ and $\beta$ 
is replaced by a new intersection $u'$ that lies on the opposite side of $\zeta$. 
The following arguments show that the number of intersections between any pair of curves within 
$\tilde{\gamma}$ does not change. If $\zeta=\gamma$, and $T$ lies in $\tilde{\gamma}$, 
then as $u'$ lies outside $\tilde{\gamma}$,
the number of intersection points in $\tilde{\gamma}$ decreases by $1$.

We now argue that all pairs of curves remain non-piercing and the number of pairwise intersections remains unchanged. 
In particular, we show that the intersection points between any pair of curves remains reverse-cyclic, 
and hence by Lemma~\ref{lem:revcyclic}, they remain non-piercing.   
Let $\delta$ be any curve in $\Gamma\,\setminus\,\{\alpha\}$. 

\smallskip\noindent
{\em Case 1. $\delta=\zeta$:} $\zeta$ and the modified $\alpha$ intersect at a new point $v$ instead of at $w$. Since $\alpha$ and $\zeta$ don't intersect at any point between $v$ and $w$ on the boundary of $T$, their reverse-cyclic sequences remain the same, and hence they remain non-piercing. 

\smallskip\noindent
{\em Case 2:} $\delta = \beta$. The modified $\alpha$ and modified $\beta$ intersect at $u'$ instead of $u$. Since $\alpha$ and $\beta$ did not have any intersections other than $u$ on the boundary of $T$, this does not change their reverse-cyclic sequences. 
Thus, they remain non-piercing. 

\smallskip\noindent
{\em Case 3:} $\delta \in \Gamma \,\setminus\, \{\alpha, \beta, \zeta\}$. If $\delta$ does not intersect $T$, the intersection points between $\alpha$ and $\delta$ does not change and therefore they remain non-piercing. Let us suppose therefore that $\delta$ intersects $T$. Its intersection with $T$ then consists of a disjoint non-intersecting set of arcs with one end-point each on the two sides of $T$ defined by $\alpha$ and $\beta$. Bypassing $T$ moves the intersection between $\delta$ and $\alpha$ from one end-point to the other on each of the arcs. However since the arcs are non-intersecting, this does not affect their reverse-cyclic sequences. Thus, they remain non-piercing.

An analogous argument shows that $\beta$ is non-piercing with respect to the other curves. 
\end{proof}

\begin{corollary}
\label{cor:minT2int}
Let $\AG$ be an arrangement of two-intersecting curves $\Gamma$ with sweep curve $\gamma$. 
If $T$ is a minimal triangle in $\tilde{\gamma}$ formed by curves $\alpha, \beta$ and $\zeta$ in $\Gamma$, then after untangling $T$, the arrangement remains two-intersecting. Furthermore, if $\zeta=\gamma$, and $T$ lies in $\tilde{\gamma}$, then
after untangling $T$, the number of intersection points in $\tilde{\gamma}$ decreases by $1$.
\end{corollary}

Let $\AG$ be an arrangement of two-intersecting curves $\Gamma$. Let $T$ be a minimal triangle on $\zeta$ formed by
$\alpha, \beta$ and $\zeta$. Let $u$ denote the vertex of $T$ defined by $\alpha$ and $\beta$.
If $\alpha$ and $\beta$ intersected at a point $u'$ so that $u,u'$ were the
vertices of a digon cell in the original arrangement, then on untangling $T$, we obtain a new minimal 
triangle $T'$ on $\zeta$, where $u'$ is the vertex of $T'$ defined by $\alpha$ and $\beta$.
See Figure~\ref{fig:doubletrianglebypass}. In such cases, we apply the minimal triangle untangling
operation again on $T'$, after which $\alpha$ and $\beta$ only intersect on the opposite side of $\zeta$.
Note that $u'$ is the vertex of $\alpha$ and $\beta$ on $T'$ and on untangling $T'$ 
the two cells adjacent to $u'$ have reference points $u'_\alpha$ and $u'_\beta$, but these are identical
to the cells with reference points $u_\alpha$ and $u_\beta$, respectively. Hence, we refer to the cells
as those with reference points $u_\alpha$ and $u_\beta$.
For ease of notation, we just refer to this operation as \emph{applying minimal triangle untangling on 
$T$ twice}.
Thus, if $\zeta=\gamma$, and $T$ was in $\tilde{\gamma}$, then on applying minimal triangle untangling twice,
in the new arrangement, $\alpha$ and $\beta$ don't intersect in $\tilde{\gamma}$. 

\begin{lemma}
\label{lem:taftert}
Let $\AG$ be an arrangement of two-intersecting curves $\Gamma$, and let $T$ be a minimal triangle 
on $\zeta$ formed by curves $\alpha, \beta$ and $\zeta$ in $\Gamma$. Then, on
applying triangle untangling on $T$ at most twice, 
the only new digon cells or triangle cells formed in the arrangement
are the ones with reference points.
$u_\alpha$ and $u_\beta$. 
\end{lemma}
\begin{proof}
The cells with reference points $u_\alpha$ and $u_\beta$ lose a vertex on untangling the minimal triangle, and hence can
potentially become digon cells or triangle cells. The cell with reference point $u$ has lost a vertex, but gained two
additional vertices, and since we apply triangle untangling twice, 
the cell has at least 4 sides. For all other cells, the number of sides
remains the same.
\end{proof}
\begin{figure}[t!]
    \centering
    \begin{subfigure}{0.3\textwidth}
        \centering
        \includegraphics[scale=0.3]{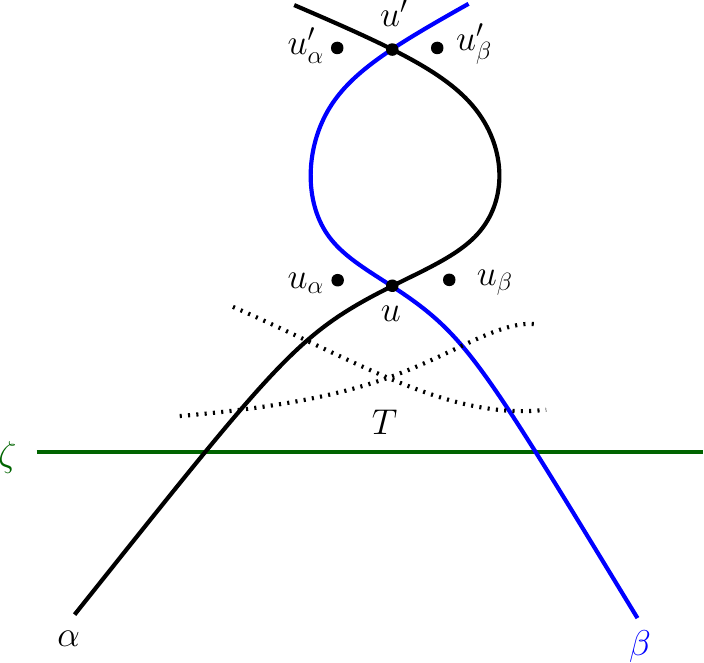}
        \label{fig:triangleuntangle_twice1}
    \end{subfigure}
    \begin{subfigure}{0.35\textwidth}
        \centering
        \includegraphics[scale=0.3]{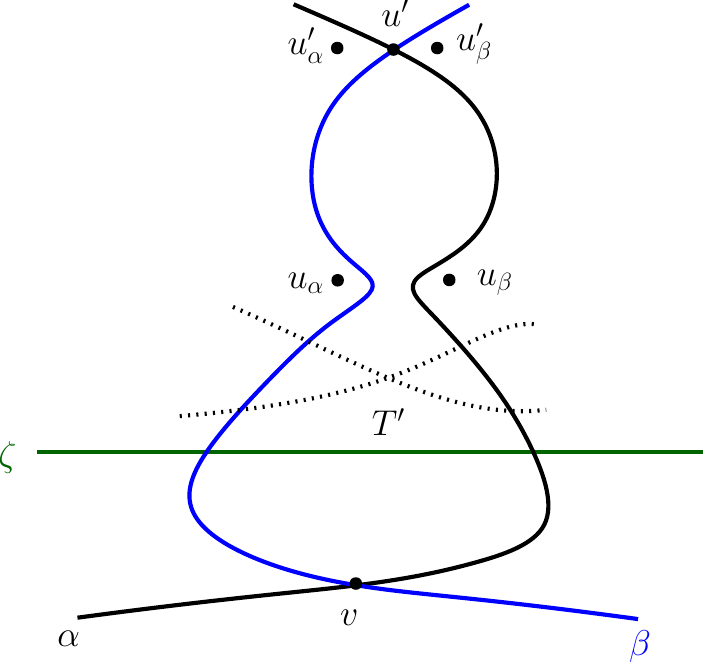}
        \label{fig:triangleuntangle_twice2}
    \end{subfigure}
    \begin{subfigure}{0.3\textwidth}
        \centering
        \includegraphics[scale=0.3]{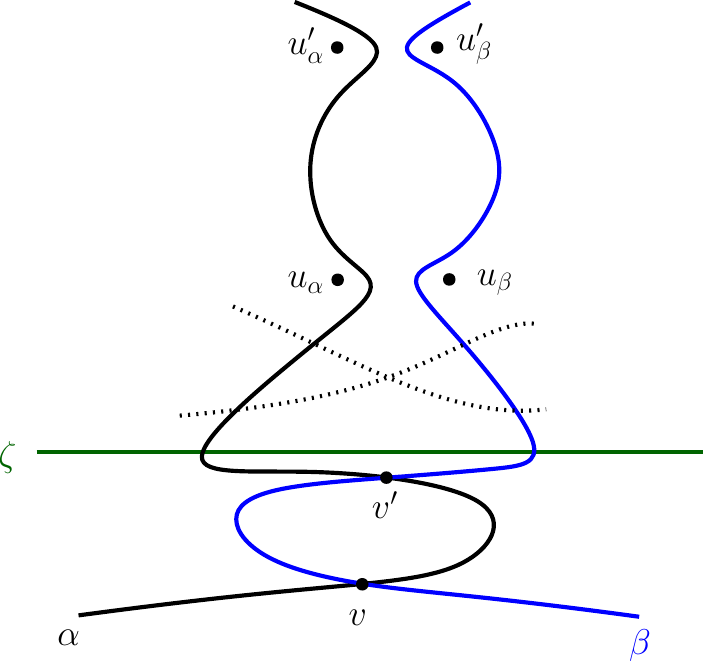}
        \label{fig:triangleuntangle_twice3}
    \end{subfigure}
    \caption{Illustration of the operation of \emph{minimal triangle untangling on $T$ twice}. The vertices $u$ and $u'$ gets mapped to $v$ and $v'$ respectively which lies to the other region defined by $\zeta$. Observe that the cell containing the reference point $u_\alpha$ also contains the reference point $u'_\alpha$, and the same is true for the cell with reference points $u_\beta$ and $u'_\beta$. Therefore, we can refer to these cells unambiguously as the ``the cell containing reference point $u_\alpha$ or
    $u_\beta$''.}
    \label{fig:doubletrianglebypass}
\end{figure}

\section{Sweeping Non-piercing Regions}
In the following, we untangle a (not-necessarily) minimal digon defined analogously to the untangling of a minimal digon (See Section~\ref{sec:basicops}). Recall that a negative lens is a digon that is not contained in regions bounded by the curves defining it. A {\em minimal negative lens} refers to a negative lens which does not contain another negative lens.
\begin{claim}
\label{claim:neglensbypass}
Let $L$ be a minimal negative lens in $\tilde\gamma$ formed by curves $\alpha$ and $\beta$. 
Let $\Gamma'$ be the modified set of curves obtained by untangling $L$ (which modifies only $\alpha$ and $\beta$). Then, $\Gamma'$ is a non-piercing set of curves and the number of vertices in $\mathcal{A}(\Gamma')$ is $2$ less than the number of vertices in $\AG$. 
\end{claim}
\begin{proof}
Let $u$ and $v$ be the intersection points of $\alpha$ and $\beta$ on the boundary of $L$ so that $\alpha$ is oriented from $u$ to $v$. Note that then, $\beta$ is oriented from $v$ to $u$. Upon untangling $L$ both $\alpha$ and $\beta$ lose the intersections $u$ and $v$ which are consecutive intersections along the two curves but appear in  reverse order along their orientations. Thus, after untangling the sequences of intersections of $\alpha$ and $\beta$ along the two curves appear in reverse-cyclic order implying that they remain non-piercing by Lemma~\ref{lem:revcyclic}. 

We now need to show that any other curve 
 $\tau\in\Gamma\,\setminus\,\{\alpha,\beta\}$ 
remains non-piercing with respect to the modified $\alpha$ and $\beta$ after untangling $L$. 
The portion of $\tau$ within $L$
consists of a collection of disjoint segments with end-points on the boundary of $L$.
If all such segments have one end-point on $\alpha$ and the other on $\beta$, 
then by an argument similar to the proof of Lemma~\ref{lem:digonbypass}, $\tau$ remains non-piercing with respect to $\alpha$ and $\beta$. Now, consider the case where a segment of $\tau$ within $L$ has two end points on $\beta$ (an analogous argument works for $\alpha$). 
 Then $\tau$ and $\beta$ form a digon $D$ in $L$.
Since $L$ is a minimal negative lens, $D$ lies outside $\tilde\beta$ and inside $\tilde{\tau}$.
Let $x$ and $y$ be the vertices of $D$ so that $\beta$ is oriented from $x$ to $y$. Then, note that
$\tau$ is also oriented from $x$ to $y$. This implies, by Lemma~\ref{lem:revcyclic}, that $\beta$ and $\tau$ do not have any other intersections apart from $x$ and $y$, and the remaining portion of $\tau$, excluding the portion on the boundary of $D$, lies in $\tilde{\beta}$. This means that $\tau$ cannot have any other intersections with the boundary of $L$. 
On untangling $L$, $\beta$ and $\tau$ do not intersect any more and are therefore non-piercing. 
The arc of $\beta$ on $L$ is replaced by the arc of $\alpha$ on $L$. 
The modified $\alpha$ and $\tau$ now have two more intersections namely $x$ and $y$ which are consecutive along both curves and appear in opposite order along their orientations implying that the sequences of intersections along them remain reverse-cyclic which implies by Lemma~\ref{lem:revcyclic}, that the two curves are non-piercing. Since no other curve is
modified, the resulting set of curves $\Gamma'$ is non-piercing and the arrangement $\mathcal{A}(\Gamma')$ has two fewer vertices than $\AG$.

As a result of untangling $L$, we lose the two vertices $u$ and $v$. For every other vertex $x$ on the portion on the boundary of $L$, defined by $\alpha$ and another curve $\mu\in\Gamma$, there
is a corresponding vertex at the same position in $\mathcal{A}(\Gamma')$ defined by the modified $\beta$ and $\mu$, and similarly
for each vertex $y$ on $L$ defined by $\beta$ and another curve $\mu\in\Gamma$, there is a
vertex $y'$ in $\mathcal{A}(\Gamma')$ defined by the modified $\alpha$ and $\mu$. All other 
vertices not on the boundary of $L$ are also vertices in $\mathcal{A}(\Gamma')$.
Since $\mathcal{A}(\Gamma')$ does not have any other vertices, $\mathcal{A}(\Gamma')$ has two fewer
vertices than $\AG$.
\end{proof}

\label{sec:nonpiercing}

\np*
\hide{
\begin{theorem}
\label{thm:mainthm}
Let $\Gamma$ be a finite non-piercing arrangement of  curves. Given any $\gamma \in \Gamma$ and a point $P \in \tilde{\gamma}$, we can sweep 
$\Gamma$ with $\gamma$ so that at any point of time during the sweep, the curves remain non-piercing and $P$ remains in $\tilde{\gamma}$. 
\end{theorem}
}
\begin{proof}
Let us call a sweeping operation valid if it does not move $P$ out of $\tilde\gamma$. 
Note that the operations of taking a new loop or vertex bypassing are always valid since we can implement them in such a way that $P$ is not moved out of $\tilde\gamma$. The only possible sweeping operations that may move $P$ out of $\tilde{\gamma}$ are therefore the operations of bypassing a digon cell or a triangle cell.

If we can always find a valid sweeping operation, 
we can continue applying them, eventually ending up in a situation where no pair of curves intersects in the interior of $\tilde{\gamma}$, while maintaining the non-piercing property (by Lemma~\ref{lem:ops}).
Thus, for contradiction, assume that there exists a set of curves with no valid sweeping operations.
Among such sets, let $\Gamma$ be a {\em non-sweepable set} (i.e., no valid sweeping operations apply) which is {\em simplest} in the following sense:  $\AG$ lexicographically minimizes $(m, n)$ where $m$ is the number of curves in $\Gamma$ that lie in the interior of $\tilde{\gamma}$ and $n$ is the number of vertices in $\Gamma$ that lie in $\tilde{\gamma}$ (including those that lie on $\gamma$). 
If no pair of curves in $\Gamma$
intersects more than twice in $\tilde{\gamma}$, then by the result of Snoeyink and Hershberger, (Theorem~\ref{thm:snoeyink}), there is at least one valid sweeping operation. 

The result of Snoeyink and Hershberger can also be applied if some curve $\alpha$ intersects $\gamma$ in more than $2$ points, but no pair of curves intersects more than twice in $\tilde{\gamma}$, since we can treat each segment of $\alpha$ in $\tilde{\gamma}$ as a distinct curve and apply their result. We may thus assume that no curve intersects $\gamma$ in more than $2$ points.

Since we assumed that no sweep operations are valid, there must be a pair of curves in $\Gamma\,\setminus\,\{\gamma\}$ that intersects $3$ or more times in $\tilde{\gamma}$, and hence form a negative lens inside $\tilde{\gamma}$. To see this, note that since the curves are non-piercing, by Lemma~\ref{lem:revcyclic}, the intersection points of any pair of curves form  
reverse-cyclic sequences. This means that they intersect alternately in lenses and negative lenses. Thus, any pair of curves intersecting more than twice
in $\tilde{\gamma}$ must form a negative lens in $\tilde{\gamma}$.

Let $L$ be a minimal negative lens in $\tilde{\gamma}$ defined by curves $\alpha$ and $\beta$
in $\Gamma$.
We can assume that $L$ is
visible from $\gamma$ (i.e., some point on the boundary of $L$ is visible from $\gamma$) as otherwise we can untangle it to obtain a simpler non-sweepable set.  We will now show that untangling $L$ yields a simpler non-sweepable set by showing that after the untangling, the modified arrangement does not contain any bypassable vertex, or sweepable digon or triangle cells, contradicting the minimality of $\Gamma$.  

Let $u$ and $v$ be the two intersection points of $\alpha$ and $\beta$ on the boundary of $L$ so that $\alpha$ is oriented from $u$ to $v$ and $\beta$ is oriented from $v$ to $u$. 
Suppose now that we untangle $L$. Note that this keeps the arrangement non-piercing by Claim~\ref{claim:neglensbypass}.
The vertices $u$ and $v$ are {\em lost} as a result of the untangling of $L$. 
We show that no vertex becomes bypassable as a result of this untangling. 
Suppose for contradiction that there was a vertex in the modified arrangement that became bypassable.
Then, its bypassability curve in the modified arrangement must lie in the modified $\tilde{\alpha}$ and
$\tilde{\beta}$, and go through the bounded region defined by $L$ in the original arrangement (before $L$ was bypassed).
This implies that either $u$ or $v$ was visible in the original arrangement which 
in turn contradicts our assumption that $\Gamma$ was non-sweepable. Figures~\ref{fig:BypassingU1} and~\ref{fig:BypassingU2}
shows this for a vertex $x_{\alpha}$ on $\alpha$.
\hide{
The only vertices that can potentially become bypassable are those that lie in $L$. 
Let $x_\alpha$ be a vertex on the boundary of $L$ and that becomes bypassable. 
We assume without loss of generality that it lies on the modified curve $\alpha$. 
Its bypassibility curve then {\em arrives} at $x_\alpha$ from the left of the modified curve $\alpha$ 
(i.e., from the interior of the new $\tilde{\alpha}$) which means that such a curve must pass arbitrarily close to either $u$ or $v$. 
Assume that it passes close to $u$ (see Figure~\ref{fig:BypassingU1}), the other case being analogous.
We can then modify it to terminate at $u$ (as shown in Figure~\ref{fig:BypassingU2}) 
so that it arrives at $u$ from the left of both $\alpha$ and $\beta$ 
(i.e., from the interior of both $\tilde{\alpha}$ and $\tilde{\beta}$) in the original arrangement. 
This shows that $u$ was bypassable before we untangled the lens $L$ - which by assumption was not the case. 
Similarly the vertices on the curve $\beta$ cannot become bypassable as a result of untangling $L$. 
An analogous argument also shows that vertices lying in the interior of $L$ 
(in the original arrangement) cannot become bypassable after $L$ is untangled.
}

\begin{figure}
\begin{subfigure}{0.45\textwidth}
    \begin{center}
        \includegraphics[scale=0.27]{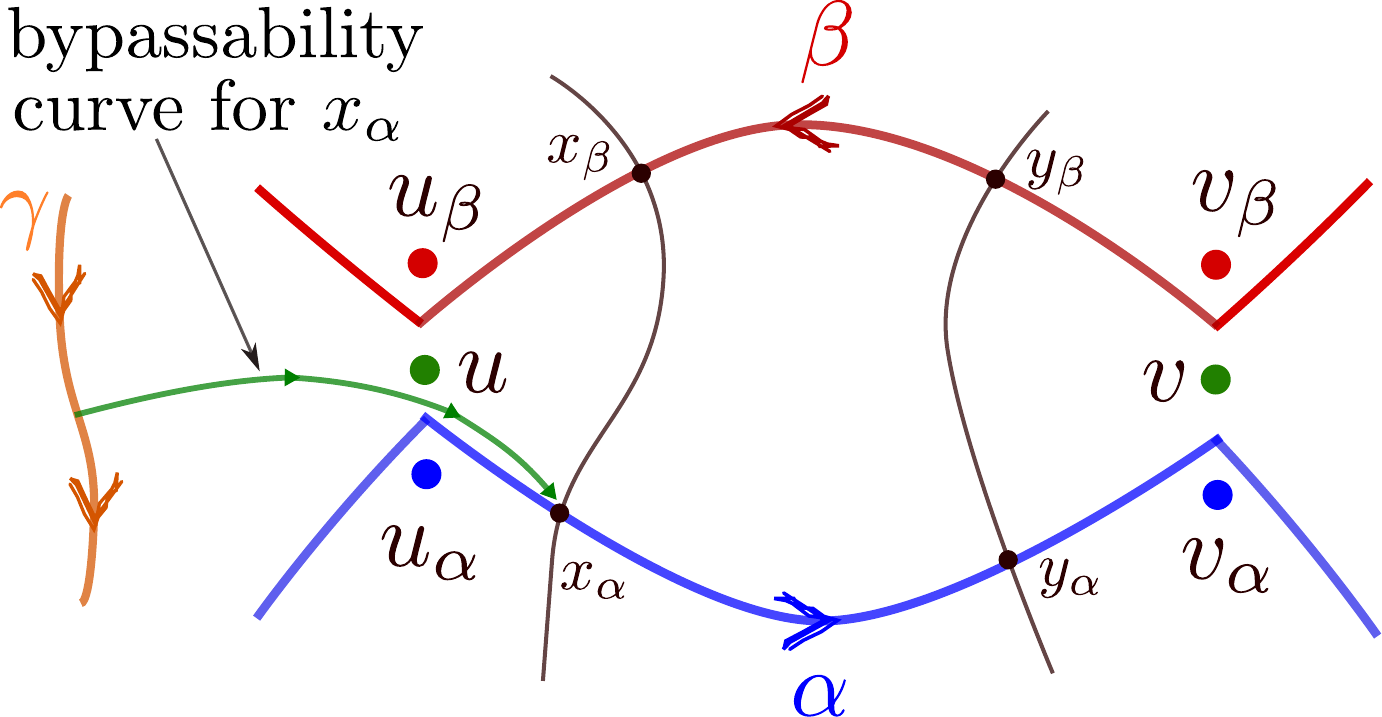}
    \subcaption{\centering Bypassability curve for $x_\alpha$.}
     \label{fig:BypassingU1}
    \end{center}
\end{subfigure}
\hspace{0.4cm}
\begin{subfigure}{0.45\textwidth}
    \begin{center}
        \includegraphics[scale=0.27]{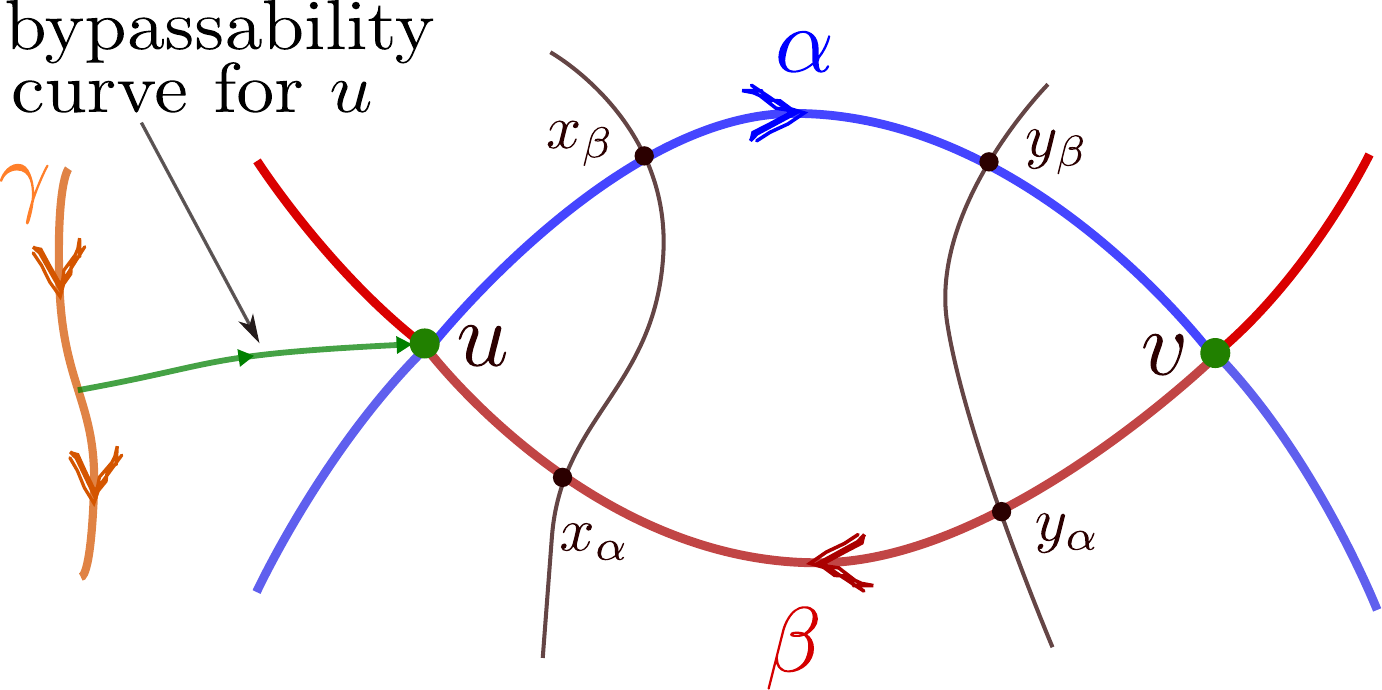}
    \subcaption{\centering Bypassability curve for $u$.}
    \label{fig:BypassingU2}
    \end{center}
\end{subfigure}
 \label{fig:BypassingU12}
 \caption{\centering If a vertex becomes bypassable after we untangle $L$, the corresponding bypassability curve can be modified to a bypassability curve for $u$ in the original arrangement.}
\end{figure}
\begin{figure}
    \centering
    \includegraphics[scale=0.27]{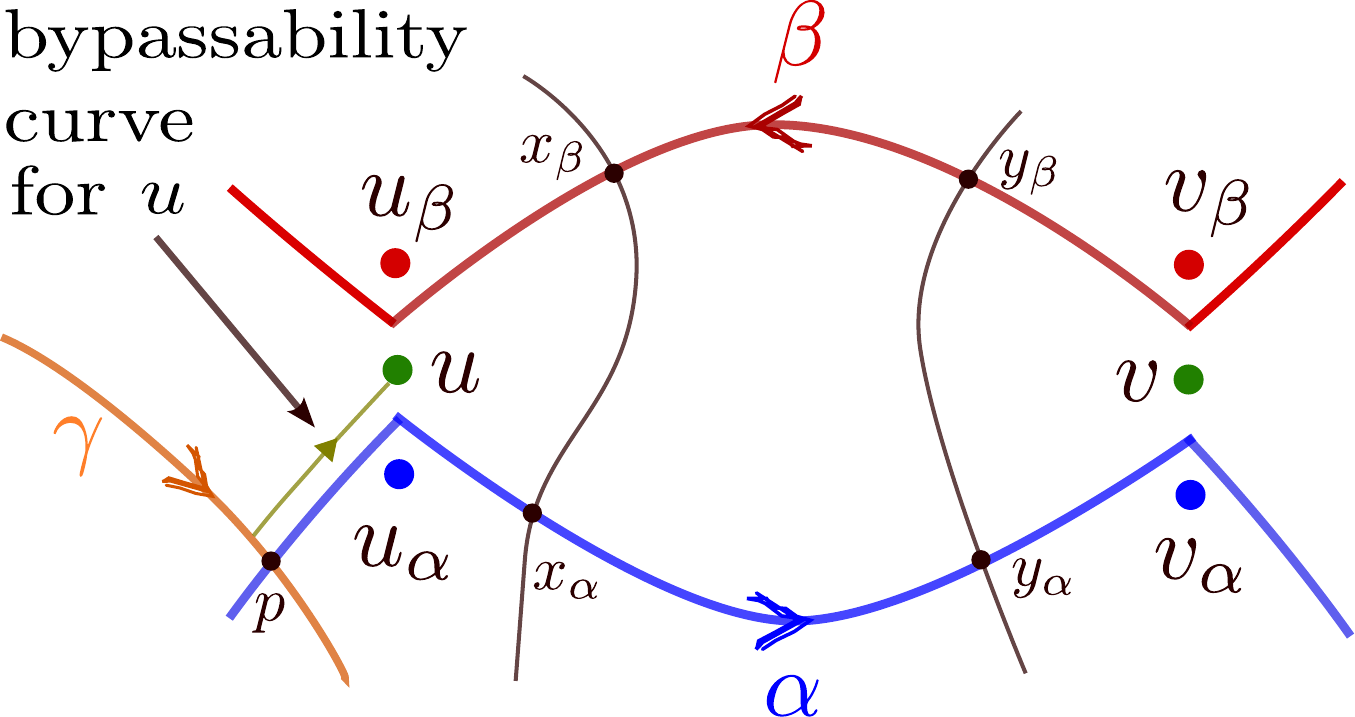}
    \caption{\centering If a cell with $u_\alpha$ is a digon cell or a triangle cell, then it was bypassable in $L$.}
    \label{fig:BypassingU3}
\end{figure}

\hide{Change arrow directions in figure 9a,b}
\hide{Add the case for vertices in $L$ formed by two other curves intersecting in $L$.}

We now show that the new cells created as a result of untangling $L$ 
cannot be digon cells or triangle cells on $\gamma$. 
The only new cells that could potentially be digon cells or triangle cells after untangling $L$ 
are those that contain one of the points $u, v, u_\alpha, u_\beta, v_\alpha$ or $v_\beta$. 
The other new cells may have $\alpha$ replaced by $\beta$ (or vice versa) 
on their boundary as result of untangling $L$ but this does not change their shape or whether they are on $\gamma$. 
Since they were not digon cells or triangle cells on $\gamma$ before, they are not digon cells or triangle cells after untangling $L$. 

The new cells containing $u$ or $u_\alpha$ cannot become a digon cell or a triangle cell on $\gamma$.
To see this consider the point $p$ which is the first intersection point of $\alpha$ and $\gamma$
that we arrive at if we start at $u$ on the original curve $\alpha$ and walk along it 
in the direction opposite to its orientation. If a new cell containing $u$ or $u_\alpha$ is a
digon cell or a triangle cell, then the point $p$ (shown in 
Figure~\ref{fig:BypassingU3}) must be a vertex of the cell and therefore in the original arrangement the arc of $\alpha$ from $p$ to $u$ along its orientation must be an edge in $\AG$. 
This implies that in $\AG$, we could construct a bypassability 
curve for $u$ as follows: start at a point on $\gamma$ arbitrarily close to $p$ that lies in the interior of $\tilde\alpha$ and follow $\alpha$ closely until $u$ and terminate there - 
thus arriving at $u$ from the interior of both $\tilde\alpha$ and $\tilde\beta$ (see Figure~\ref{fig:BypassingU3}). 
This implies that $u$ was bypassable in the original arrangement contradicting our assumptions.
By symmetry, the cells containing any of the other points $u_\beta, v, v_\alpha$ and  $v_\beta$ also cannot be digon cells or triangle cells.
This shows that $u$ must have been bypassable before we untangled $L$, contradicting our assumption that
$\AG$ was a {\em simplest} non-sweepable arrangement. 
\end{proof}

{\bf Remark.} While our result states that we can sweep by continuously {\em shrinking} $\tilde\gamma$ to a specified point $P$ in $\tilde{\gamma}$, the proof can be easily modified to show that we can 
sweep by {\em expanding} $\tilde\gamma$
so that a specified point $Q$ remains outside $\tilde{\gamma}$. If $Q$ is chosen to be a point at infinity, then this shows that we can sweep by expanding $\gamma$ to infinity
retaining the property that the curves are non-piercing.  The only change
required in the proof is that a vertex $q$ defined by curves $\alpha$ and $\beta$ is said to be bypassable if there is a visibility curve $\tau$ from a point on $\gamma$ to $q$ s.t. the interior of $\tau$ lies outside $\tilde{\gamma}, \tilde{\alpha}$, and $\tilde{\beta}$.

\section{Alternate Proof of Snoeyink-Hershberger Result}
\label{sec:snoeyinkthm}
In this section, we give an alternate proof of the result of Snoeyink and Hershberger 
stated in Theorem~\ref{thm:snoeyink}. While our proof
follows their general framework, we partition
into cases differently and the analysis of the cases is relatively simpler. 

Let $\Gamma$ be a set of pseudocircles and let $P$ be a specified point in $\tilde{\gamma}$.
Our goal is to sweep $\AG$ with a sweep curve $\gamma \in \Gamma$ so that the curves remain a set of pseudocircles throughout the process (except a finite set of instants in time) and $P$ remains inside $\tilde\gamma$ 
until we have modified $\gamma$ so that no curve in 
$\Gamma\,\setminus\,\gamma$ intersects $\tilde{\gamma}$. 
We do not use the operation of bypassing a visible vertex, as this violates the
property that any pair of the curves intersect at $0$ or $2$ points. 
We restrict ourselves to the remaining sweeping operations, namely: 
$(i)$ take a new loop,
$(ii)$ bypass a digon cell, and $(iii)$ bypass a triangle cell. 
By Corollary~\ref{cor:two}, the curves remain a set of pseudocircles on applying any one of these operations.

A sweeping operation is said to be valid if it does not move $P$ out of $\tilde\gamma$. 
A cell in the arrangement is said to be \emph{sweepable} if it is a digon cell or a triangle cell on $\gamma$
that does not contain $P$. Otherwise, it is non-sweepable.
If at any point in time, there is a sweepable cell, i.e., there is a valid sweeping operation,
 we say that the arrangement is \emph{sweepable}. Otherwise,
we say that it is non-sweepable.

We will show that any set $\Gamma$ of pseudocircles with a sweep curve $\gamma \in \Gamma$ and
a point $P\in\tilde{\gamma}$ is sweepable. 
For contradiction, assume that this is not true, and among arrangements that are non-sweepable consider the {\em simplest} one in the following sense: it lexicographically minimizes the tuple $(n, \ell)$ where  $n = |\Gamma|$ and $\ell$ is the number of intersection points (among curves in $\Gamma$) lying
in $\tilde{\gamma}$ (including those on $\gamma$). 

\begin{observation}
\label{obs:nonsweepable}
In a simplest non-sweepable arrangement, 
all curves must intersect the sweep curve $\gamma$.
\end{observation}

To see the above, note that if there is a curve $\alpha$ that does not intersect $\gamma$, then it must lie entirely in the interior of $\tilde\gamma$ or the exterior of $\tilde\gamma$. If it lies in the exterior of $\tilde\gamma$, then removing it yields a simpler non-sweepable arrangement. Similarly, if it lies in the interior of $\tilde{\gamma}$, it cannot be visible from $\gamma$ as otherwise the there is a valid sweep operation (namely `taking a new loop') implying that the arrangement is sweepable. On the other hand, if it is not visible from $\gamma$, removing it yields a simpler non-sweepable arrangement. Thus, the existence of any such curve $\alpha$ yields a contradiction to the assumption that we have a simplest non-sweepable arrangement.\\

{\em Broad idea.} In the restricted setting where each pair of curves in $\Gamma\,\setminus\,\{\gamma\}$
intersects at most once inside $\tilde{\gamma}$, we can show that the arrangement is sweepable. Therefore, a simplest
non-sweepable arrangement must have at least one pair of curves that intersect twice (i.e., they form a digon) in $\tilde{\gamma}$.
In this case, we carefully modify the arrangement by untangling a minimal digon or a minimal triangle (defined in Section~\ref{sec:basicops})
to obtain a simpler non-sweepable arrangement, thus arriving at a contradiction.\\

We start with the following definition that is borrowed from the paper of Snoeyink and Hershberger~\cite{SnoeyinkH90}.
Let $T$ be a triangle defined by curves $\alpha$ and $\beta$ on $\gamma$.
We say that $T$ is a \emph{half-triangle} with edge $\alpha$ if the interior of the side of the triangle defined by $\alpha$ is not intersected by any curve in $\Gamma \,\setminus\, \{\alpha\}$. Figure~\ref{fig:halftriangle} shows a half-triangle with edge $\alpha$.

\begin{figure}[ht!]
\begin{center}
\includegraphics[width=2in]{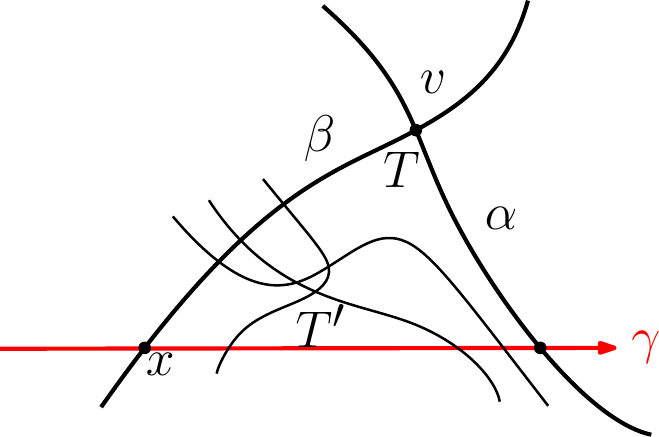}
\caption{\centering $T$ is half-triangle with edge $\alpha$. The triangle cell $T'$ lies in $T$.}
\label{fig:halftriangle}
\end{center}
\end{figure}

\begin{restatable}{lem}{onceintersect}
\label{lem:onceintersect}
Let  $\Gamma$ be a set of pseudocircles with sweep curve $\gamma\in\Gamma$ s.t. all curves in $\Gamma\,\setminus\,\{\gamma\}$ intersect $\gamma$, and
any pair of curves in $\Gamma\,\setminus\,\{\gamma\}$ 
intersect at most once in $\tilde{\gamma}$. Then, there is a valid sweepable cell on $\gamma$, i.e., there
is either a digon cell or a triangle cell on $\gamma$ that does not contain $P$.
\end{restatable}

In order to prove the above lemma, we need the following claim which though not explicitly stated is also proved as part of Lemma 3.1 of~\cite{SnoeyinkH90}.
\begin{claim}
\label{claim:halftriangle}
If curves $\alpha$ and $\beta$ form a half-triangle $T$ on $\gamma$ with edge $\alpha$, so that each curve in $\Gamma \,\setminus\, \{\gamma\}$ has at most one intersection with the interior of the side of $T$ on $\gamma$, then 
there is a triangle cell $T' \subseteq T$ on $\gamma$. 
\end{claim}

\begin{proof}
We prove by induction on the number of curves intersecting the side of $T$ on $\gamma$. If no curve intersects the interior of this side,
since the curves pairwise intersect in at most one point in $\tilde{\gamma}$, it follows that $T$ is a triangle-cell.
Assume that the claim holds for any half-triangle with less than $k$ curves (where $k \geq 1$) intersecting the interior of the side of $T$ on $\gamma$.
Suppose now that there are $k$ curves intersecting the interior of the side of $T$ on $\gamma$. 

Let $x$ denote the intersection of $\beta$ and $\gamma$ and let  
$v$ denote the intersection of $\alpha$ and $\beta$ on the boundary of $T$. 
For every
curve $\delta'$ intersecting the interior of the side $T$ on $\gamma$, $\delta'$ intersects the side at one point,
but $\delta'$ and $\gamma$ intersect at two points. Therefore, since $T$ is a half-triangle with edge $\alpha$,
$\delta'$ intersects the side of $T$ on $\beta$. Walking from $x$ to $v$ along $\beta$, 
let $\delta$ be the first curve
intersecting $\beta$. The curve $\delta$ intersects the interior of the side of $T$ on $\gamma$. 
\hide{
Since every curve has two intersection points on $\gamma$ and at most one intersection point on the side of $T$ on $\gamma$, and $k\ge 1$, there must be such a curve $\delta$.
Since the curves are $1$-intersecting in $\tilde{\gamma}$, and $\alpha$ is an edge of the half-triangle $T$, 
this implies that $\delta$ intersects the side of $T$ on $\gamma$.}
Now, $\beta$ and $\delta$ form a half-triangle $T''$ on $\gamma$ with edge $\beta$ and s.t. 
less than $k$ curves intersect the interior of the side of $T''$ on $\gamma$. Since every curve in $\Gamma$ has at most one intersection point on the side of $T$ on $\gamma$, the same holds for $T''$. Hence,
by the inductive hypothesis, there is a triangle cell $T'$ on $\gamma$ that lies in $T''$. Since $T''$ lies in $T$, $T'$
is the claimed triangle cell in $T$. 
\end{proof}

\begin{proof}[Proof of Lemma~\ref{lem:onceintersect}]
Each curve $\alpha\in\Gamma$ has two intersections with $\gamma$ which split $\gamma$ into two arcs. One of these arcs along 
with $\alpha$ bounds a portion of $\tilde\gamma$ that does not contain $P$. We denote this arc by $I_{\alpha}$.
The containment order on the arcs induces a partial order $\prec$ on the curves in 
$\Gamma\,\setminus\,\{\gamma\}$ ($\alpha\prec\beta\Leftrightarrow I_{\alpha}\subseteq I_{\beta}$).
Let $\alpha$ be a minimal curve with respect to $\prec$.
If the digon $D$ defined by $\alpha$ and $\gamma$ is a digon cell, then we are done. 
Otherwise, let $i$ and $j$ be the vertices of $D$ so that $I_\alpha$ is the arc from $i$ to $j$ along $\gamma$ in the direction of $\gamma$'s orientation.
Let $\beta$ be the
first curve intersecting $\alpha$ when following the arc of $\alpha$ on the boundary of $D$ from $i$ to $j$.
Since each curve in $\Gamma\,\setminus\,\{\gamma\}$ intersects $\gamma$ twice, and pairwise intersect at most once in $\tilde{\gamma}$, $\beta$ intersects both $\alpha$ and $\gamma$ exactly once on the boundary of $D$.
Thus, $\alpha,\beta$ and $\gamma$ form a half-triangle $T$ with edge $\alpha$. Note that $T$ lies within $D$, and therefore does not contain $P$. By Claim~\ref{claim:halftriangle}, $T$ contains a triangle cell, and this completes the proof.
\end{proof}


The proof of the main result in this section (Theorem~\ref{thm:snoeyink}) follows the structure of Theorem~\ref{thm:mainthm}, 
i.e., the proof for non-piercing curves.
Unlike in the non-piercing case however, since the curves pairwise intersect in $0$ or $2$ points,
no pair of curves form a negative lens in $\tilde\gamma$. 

Let $\Gamma$ be a set of pseudocircles so that $\AG$ is the simplest non-sweepable arrangement and let $L$ be a minimal digon (lens or lune) in $\tilde\gamma$
defined by curves $\alpha,\beta\in\Gamma$. Since $\AG$ is assumed to be non-sweepable, $\gamma \notin \{\alpha, \beta\}$.
We split the proof into two cases: In the first case, there is a curve $\delta\in\Gamma\,\setminus\,\{\alpha,\beta\}$ intersecting $L$.
In this case, we show that we can modify the arrangement to obtain a simpler non-sweepable arrangement, thus arriving
at a contradiction. If no curve in $\Gamma\,\setminus\,\{\alpha,\beta\}$ intersects $L$, then we split the proof
further into sub-cases: since $\AG$ was assumed to be a simplest non-sweepable arrangement, untangling $L$ results in
a sweepable arrangement. Thus, there is either a sweepable digon cell or a triangle cell on $\gamma$ (i.e.,not containing $P$). If 
a newly created cell is a sweepable digon cell on $\gamma$,
we show that removing the other curve forming the digon cell
results in a simpler non-sweepable arrangement. 
If untangling $L$ creates a sweepable triangle cell, 
we show that a more elaborate modification yields a simpler non-sweepable arrangement.
Thus, in all cases, we obtain a contradiction, 
and hence there is always either a sweepable digon cell or a triangle cell on $\gamma$.

\begin{theorem}[\cite{SnoeyinkH90}]
\label{thm:snoeyink}
Given any set of pseudocircles $\Gamma$, a curve $\gamma\in\Gamma$ 
and a point $P\in\tilde{\gamma}$, $\Gamma$ can be swept by $\gamma$
using the following operations:  i) bypassing a digon cell, ii) bypassing a triangle cell and iii) taking a new loop, so that at any point in time during the sweep, the curves remain pseudocircles (except for a finite set of points in time), and $P$ lies in the interior of $\tilde\gamma$.
\end{theorem}
\begin{proof}
Suppose the statement is not true. 
Consider a simplest non-sweepable set of pseudocircles $\Gamma$ with sweep curve $\gamma\in\Gamma$ i.e., the simplest arrangement for which the statement of the theorem does not hold.
By Observation~\ref{obs:nonsweepable}, every curve intersects $\gamma$.
If the curves in $\Gamma\,\setminus\,\{\gamma\}$ pairwise intersect at most once in $\tilde{\gamma}$, then by 
Lemma~\ref{lem:onceintersect}, the arrangement is sweepable.  
The assumption that we have the simplest non-sweepable arrangement thus implies the following:
 $(i)$ there is a pair of curves in $\Gamma\,\setminus\,\{\gamma\}$ intersecting twice in $\tilde{\gamma}$, 
i.e., they form a digon in $\tilde\gamma$, and $(ii)$ there is either 
no digon cell or triangle cell on $\gamma$, or there is exactly one
digon cell or triangle cell on $\gamma$ and this cell contains $P$.

Let $L$ be a minimal digon contained in $\tilde\gamma$ and suppose that $L$ 
is defined by the curves $\alpha, \beta\in\Gamma\,\setminus\,\{\gamma\}$ and has vertices $u$ and $v$. 
$L$ must be visible from $\gamma$, as otherwise we can untangle it 
(by Lemma~\ref{lem:digonbypass}, the arrangement remains non-piercing) and obtain  a simpler non-sweepable arrangement. 
Let $a_1$ and $a_2$ be the intersection points of $\alpha$ with $\gamma$ s.t. the points $a_2, v, u, a_1$ lie
in cyclic order along $\alpha$ (not necessary along $\alpha$'s orientation). 
Similarly, let $b_1$ and $b_2$ be the intersection points of $\beta$ and $\gamma$ 
so that $b_1, u, v, b_2$ appear in cyclic order along $\beta$ (not necessary along $\beta$'s orientation).  See Figure~\ref{fig:sn1}. 
Let $u_\alpha, u_\beta, v_\alpha$ and $v_\beta$ be the reference points for untangling the digon $L$. 

\begin{figure}[ht!]
\centering
 \begin{subfigure}{0.5\textwidth}
\begin{center}
 \includegraphics[scale=.6]{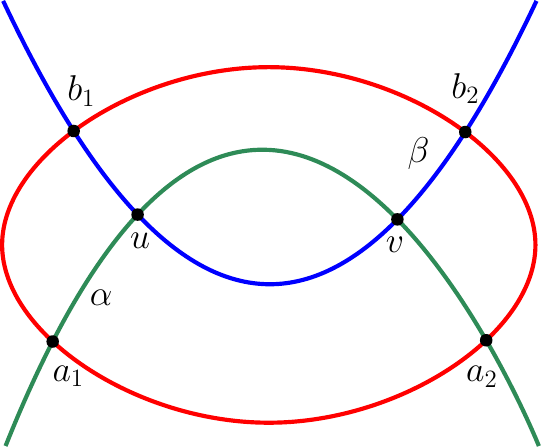}
\subcaption{\centering The curves $\alpha$ and $\beta$ form the minimal digon $L$.}
\label{fig:sn1}
\end{center}
\end{subfigure}
\begin{subfigure}{0.5\textwidth}
\begin{center}
 \includegraphics[scale=0.6]{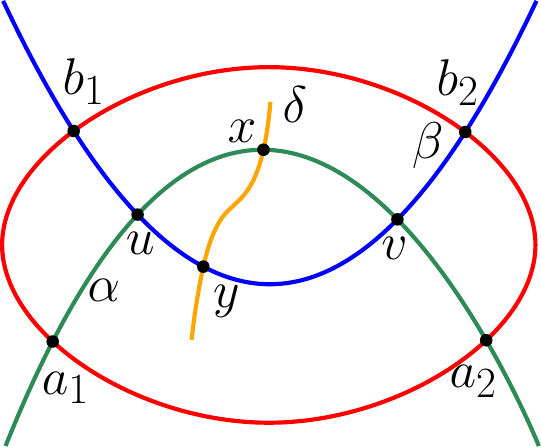}
    \subcaption{\centering $\delta$ is the first curve intersecting $L$ when walking on $\beta$ from $u$ to $v$.}
       \label{fig:12a}
    \end{center}
 \end{subfigure}
\caption{\centering  $L$ and $\delta$}
 \label{fig:sn2}
 \end{figure}
 
Untangling $L$ should create a sweepable cell 
as otherwise the simpler arrangement obtained on untangling $L$ remains non-sweepable contradicting 
the minimality of $\Gamma$. By Lemma~\ref{lem:lenstriangle}, 
the sweepable cell created must contain one of the reference points $u_\alpha, u_\beta, v_\alpha$ or $v_\beta$. 
By symmetry, we can assume without loss of generality that the cell containing $u_\alpha$ is  sweepable. 

We now split the proof into two cases depending on whether $L$ is \emph{non-empty}, i.e., there is a curve $\delta\in\Gamma\,\setminus\,\{\alpha,\beta,\gamma\}$ that intersects $L$, 
or is \emph{empty}, i.e., there is no such curve. 

\smallskip\noindent
{\bf Case 1: $L$ is non-empty.}

Since there is at least one curve intersecting $L$, let $\delta$ be the curve that has 
the first intersection $y$ with $\beta$  
as we walk on $\beta$ from $u$ to $v$. 
The intersection of $\delta$ with $L$ consists of disjoint segments with end-points on the boundary of $L$. 
Since $L$ is a minimal digon, each such segment has one end-point on $\alpha$ and one on $\beta$.
Let $x$ be the other end-point of the segment whose one end-point is $y$. Note that $x$ lies on $\alpha$. 
See Figure~\ref{fig:12a}. Let $d_1, d_2$ be the intersections of $\delta$ and $\gamma$ 
so that $d_1, x, y, d_2$ appear in cyclic order 
along $\delta$. We do not show $d_1$ and $d_2$ in Figure~\ref{fig:12a} since at this point it is not clear where they lie relative to the other points. 

\begin{figure}[!ht]
\begin{subfigure}{0.3\textwidth}
\begin{center}
 \includegraphics[scale=.38]{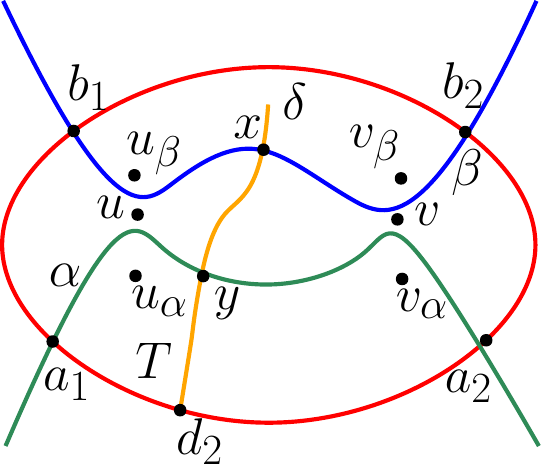}
\subcaption{\centering A triangle $T$ formed on bypassing $L$.}
\label{fig:bpd1d2a}
\end{center}
\end{subfigure}\hspace{0.3cm}
\begin{subfigure}{0.3\textwidth}
\begin{center}
 \includegraphics[scale=0.38]{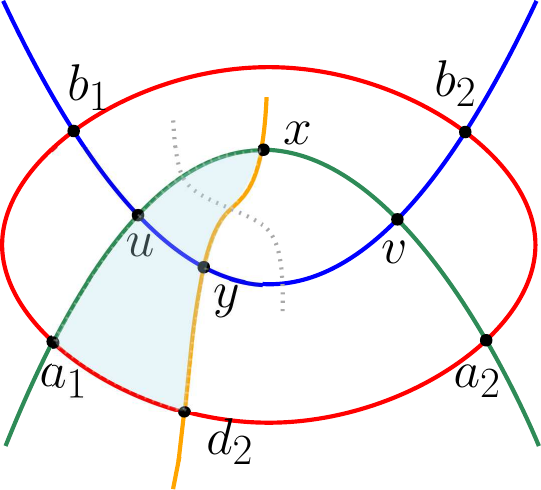}
    \subcaption{\centering The minimal triangle shaded.}
       \label{fig:bpd1d2b}
    \end{center}
 \end{subfigure}\hspace{0.3cm}
 \begin{subfigure}{0.3\textwidth}
\begin{center}
 \includegraphics[scale=0.38]{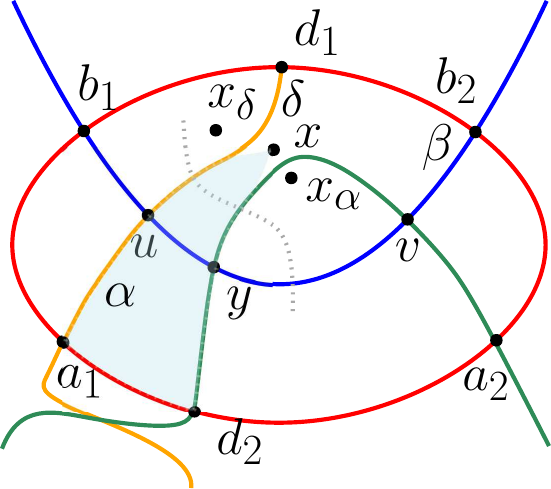}
    \subcaption{\centering Bypassing the minimal triangle.}
   \label{fig:bpd1d2c}
       \end{center}
 \end{subfigure}\hspace{0.3cm}
 \caption{\centering The operations when $L$ is non-empty.}
 \label{fig:bpd1d2}
 \end{figure}
 
\begin{claim}
\label{claim:deltagood}
The intersections of $\delta$ with any of the other curves in $\Gamma\,\setminus\,\gamma$
within $\tilde{\gamma}$ lie in $L$. In other words, there are no intersection points on 
$\delta$ that lie in the interior of $\tilde{\gamma}$ but outside $L$. 
\end{claim}

\begin{proof}
\hide{
Consider the initial arrangement where $L$ has not yet been bypassed as shown in Figure~\ref{fig:12a}. 
}
The cell $C$ containing $u_\alpha$ after untangling $L$ cannot be a digon cell on $\gamma$ 
since such a digon cell must be the digon defined by 
$\alpha$ and $\gamma$ and containing $u_\alpha$. However that digon contains the point 
$y$ and cannot be a cell in the arrangement. 
$C$ can however be a triangle cell on $\gamma$ and in this case it is defined by the curves $\alpha, \gamma$ and $\delta$ 
and has vertices $a_1, y$ and $d_2$ (Figure~\ref{fig:bpd1d2a}). 
Note that this requires that $d_2$ lies on the arc of $\gamma$ between $a_1$ and $a_2$ along $\gamma$'s orientation.
Also note that if this is the case, then in the original arrangement (before untangling $L$ - see Figure~\ref{fig:12a}), 
there are no intersection points between $a_1$ and $u$ on $\alpha$, and between $u$ and $y$ on $\beta$. 
Since $L$ is a minimal digon, any curve intersecting $L$ is composed of disjoint segments, where each segment
has a vertex on the segment of $\alpha$ between $u$ and $v$, and a vertex on the segment of $\beta$ between $u$ and $v$.
The minimality of $L$ also implies that any two curves intersecting $L$ intersect at most once inside $L$. Now, 
by the choice of $\delta$, any segment of a curve intersecting $L$ between $u$ and $x$ has its other end-point on $\beta$ between $y$ and $v$ (as shown in Figure~\ref{fig:bpd1d2b}).
This implies that triangle $T$ 
with vertices $x, a_1$ and $d_2$ is a minimal triangle on $\gamma$ in the original arrangement (Figure~\ref{fig:bpd1d2b}).
Therefore, by Lemma~\ref{lem:trianglebypass} we can untangle $T$, 
so that the resulting set of curves (Figure~\ref{fig:bpd1d2c}) are still pseudocircles and their arrangement  
has one less vertex (since we lose the intersection $x$) inside $\tilde{\gamma}$. Since the new arrangement is simpler, 
it must be sweepable which means there must be a digon cell or a triangle cell on $\gamma$ 
not containing $P$ created as a result of untangling $T$. By Lemma~\ref{lem:taftert},
only the cells containing the reference points $x_\alpha$ or $x_\delta$ can possibly be a digon cell or a triangle cell on $\gamma$. However, the
cell containing $x_\alpha$ cannot be on $\gamma$ as it is contained in a digon formed by the modified $\alpha$, and $\beta$ that lies in the interior of $\title{\gamma}$. Therefore, only the cell containing
$x_\delta$ could be a digon cell or a triangle cell on $\gamma$. However, this cell can lie on $\gamma$  only if $d_1$ lies on the arc of $\gamma$ between $b_2$ and $b_1$ (see Figure~\ref{fig:bpd1d2c}).
In particular, this implies that the interior of the arc of $\delta$ between $x$ and $d_1$ 
is not intersected by any curve in the original arrangement $\Gamma$. Since we assumed that the cell containing $u_\alpha$ is a triangle, the portion of $\delta$ between $y$ and $d_2$ is also not intersected by any other curve. 
This implies that the intersection of $\delta$ with any other curve in the original arrangement lies within the digon $L$, 
thus completing the proof. 
\end{proof}

Since the original arrangement of $\Gamma$ (Figure~\ref{fig:12a}) was not sweepable, all cells on $\gamma$ apart from the cell containing $P$ must have at least 4 vertices. Now, consider the arrangement of $\Gamma'=\Gamma\,\setminus\,\{\delta\}$ obtained by removing the curve $\delta$.
We claim that $\Gamma'$ remains non-sweepable. By Claim~\ref{claim:deltagood} the arc of $\delta$ between $d_2$ and $y$ separates two cells $C_1$ and $C_2$. Either one of them contains the point $P$ (which makes it non-sweepable), or both cells have at least four vertices.
Thus, the new cell created by merging these two cells upon removal of $\delta$ is non-sweepable.
An analogous statement holds for the arc of $\delta$ between $d_1$ and $x$.
Finally note that the remaining new cells created as a result of removing the portion of $\delta$ between $x$ and $y$ are not cells on $\gamma$ since they lie within the digon $L$. 
Thus, $\Gamma'$ is a simpler non-sweepable arrangement contradicting the minimality of $\Gamma$.

Before we proceed with the rest of the proof, we need one more technical tool that we describe next.

\smallskip\noindent{\bf Parallel chords in a Digon:}
Let $R$ be a digon bounded by the sweep curve $\gamma$ and another curve $\mu$ with vertices $b_1$ and $b_2$. 
A {\em proper chord} $\lambda$ of $R$ is an arc of another curve $\delta$ which intersects the boundary of $R$ 
exactly once on $\mu$ and once on $\gamma$. 
We say that a sequence of proper chords $\lambda_1, \cdots, \lambda_k$ are parallel 
if for $i=1,\ldots, k-1$, 
the region of $R$ between consecutive chords $\lambda_i$ and $\lambda_{i+1}$ is a four-sided cell in the arrangement that does not contain $P$.
For each $i$, let $\delta_i$ be the curve containing $\lambda_i$ and let $\Delta_\mu$ denote the sequence $\Delta_{\mu}=(\delta_1,\ldots, \delta_k)$. See Figure~\ref{fig:parallelTriangles}.

\begin{lemma}
\label{lem:parallel}
Let $\Delta_\mu = (\delta_1, \cdots, \delta_k)$ be as defined above. 
It is possible to modify the curves in $\Delta_\mu$, possibly discarding some of them, so that i) the modified arrangement is as simple, or simpler than $\Gamma$, ii) In the modified arrangement,
the only new sweepable digon cells or triangle cells created in $\tilde\gamma$ have either $b_1$ or $b_2$ as a vertex. 
\end{lemma}
\begin{proof}
If there are no digon cells or triangle cells on $\mu$, then there is nothing to do. Suppose there is a digon cell on 
$\mu$ formed by curve $\delta_i$. Then, $\Gamma\,\setminus\,\{\delta_i\}$ is a simpler arrangement with fewer
digon cells or triangle cells on $\mu$. Since $\delta_i$ forms a digon cell on $\mu$, it forms a side of a cell
$C$ in $\tilde\gamma\,\setminus\, R$. 
Removing $\delta_i$ therefore decreases the number of sides in
$C$ by 1, and does not affect any other cell in $\tilde\gamma\,\setminus\, R$. Hence, only $C$ can become a triangle.

If $\delta_i$ and $\delta_{i+1}$ form a triangle cell on $\mu$, let $p$ be the vertex defined by
$\delta_i$ and $\delta_{i+1}$ that forms a vertex of this triangle cell. 
Then, the triangle $T$ formed by $\delta_i$
and $\delta_{i+1}$ on $\gamma$ is a minimal triangle with base $\gamma$. 
We apply a minimal triangle untangling on $T$ at most twice, and by Corollary~\ref{cor:minT2int}, 
we obtain a simpler two-intersecting arrangement. 
By Lemma~\ref{lem:taftert}, the only cells
with reference points $u_{\delta_i}$ or $u_{\delta_{i+1}}$ 
can become triangle cells. 
Therefore, the
number of triangle cells on $\mu$ increases by 2, and no new triangle cells are created on $\gamma$. 
If $\delta_i$
is not the first or last curve in the sequence $\Delta$, then consider the two cells with reference points $u_{\delta_i}$
or $u_{\delta_{i+1}}$. If either of the cells is a cell on $\gamma$, then that cell has at least 4 sides: one side on 
$\delta_i$ or $\delta_{i+1}$, one side on $\mu$, one side on $\delta_{i'}$ (where
$i'<i$ or $i'>i+1$) and one side on $\gamma$.

Every such operation yields a simpler arrangement where either the number of curves in $\Delta$
decreases, or the number of intersections between the curves $\Delta$ decreases. Therefore, after a finite
number of iterations, there are no triangle cells or digon cells on $\mu$ except possibly the ones with
$b_1$ or $b_2$ as a vertex. 
\end{proof}

Figures~\ref{fig:parallelTriangles} show the modification
of the curves in $\Delta_\mu$.

\begin{figure}[ht!]
\centering
\begin{subfigure}[!htbp]{0.5\textwidth}
\begin{center}
 \includegraphics[scale=.55]{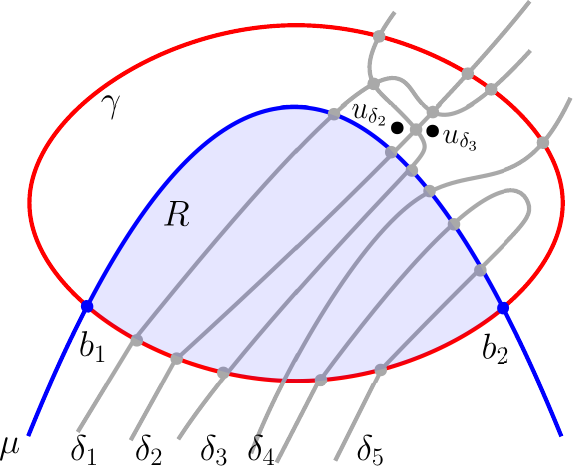}
\subcaption{\centering Parallel curves on $\mu$.}
\end{center}
\end{subfigure}
\begin{subfigure}[!htbp]{0.5\textwidth}
\begin{center}
 \includegraphics[scale=.55]{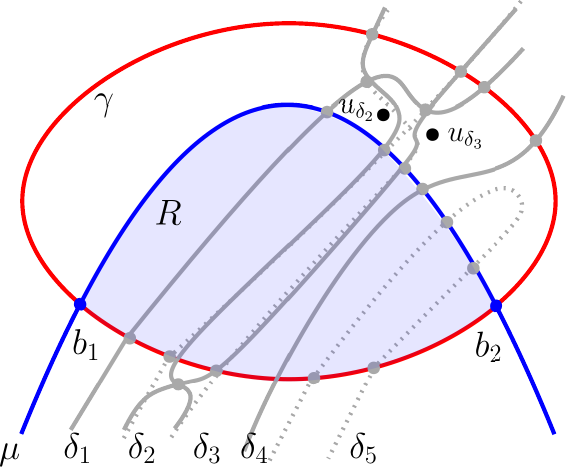}
\subcaption{\centering Triangle bypassing on $\delta_2, \delta_3$, and removing $\delta_5$.}
\end{center}
\end{subfigure}
\caption{\centering Parallel curves on $\mu$, and operations to remove triangles and digons on $\mu$.}
 \label{fig:parallelTriangles}
 \end{figure}

\smallskip\noindent
{\bf Case 2. $L$ is empty.}

This is the more complicated case.
Our plan is to again obtain a contradiction by finding a simpler non-sweepable arrangement.
Recall that our assumption that $\Gamma$ is a simplest non-sweepable arrangement implies that untangling $L$ creates a sweepable 
cell on $\gamma$. 
Let $D_u$ be the cell in the original arrangement $\Gamma$ contained in $\tilde\gamma$ with $u$ as a vertex,
and having portions of the arcs $[a_1, u]$ and
$[b_1, u]$ as two sides. Similarly, let $D_v$ be the cell containing vertex $v$, and having portions of the arcs $[a_2, v]$
and $[b_2,v]$ as sides.

\begin{figure}[ht!]
\centering
\begin{subfigure}{0.4\textwidth}
\begin{center}
 \includegraphics[scale=.6]{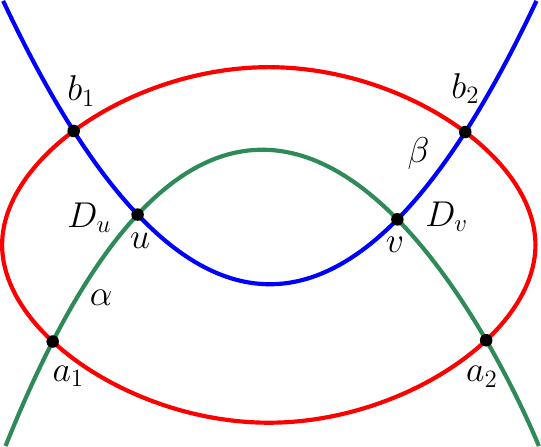}
\subcaption{\centering The cells $D_u$ and $D_v$ adjacent to the vertices $u$ and $v$.}
\label{fig:dudv}
\end{center}
\end{subfigure}\hspace{1.2cm}
\begin{subfigure}{0.4\textwidth}
\begin{center}
 \includegraphics[scale=0.6]{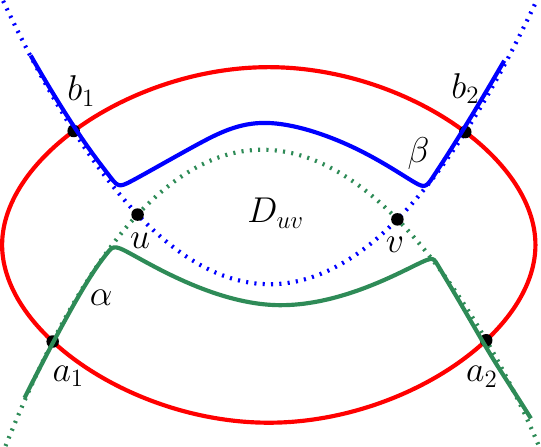}
    \subcaption{\centering The cell $D_{uv}$ obtained after untangling the empty digon $L$. }
       \label{fig:duv}
    \end{center}
 \end{subfigure}
 \caption{\centering  The cell $D_{uv}$
 is not a bypassable cell.}
 \label{fig:dudvduv}
 \end{figure}

Now suppose that we untangle $L$ to obtain a new arrangement. 
In this arrangement, consider the cell $D_{uv}$ that contains the points $u$ and $v$ in its interior.
See Figures~\ref{fig:dudvduv}.

By Lemma~\ref{lem:lenstriangle}, the only digon cells or triangle cells on $\gamma$ can be the cells that contain one of the
reference points $u_{\alpha}, u_\beta, v_\alpha$, or $v_\beta$. This in particular implies that $D_{uv}$ is not a digon cell or a triangle cell on $\gamma$.
Since $L$ was empty, there is a single cell containing $u_\alpha$ and $v_\alpha$. Let us call this cell $C_\alpha$.
Similarly, let $C_\beta$ denote the cell containing the reference points $u_\beta$ and $v_\beta$. 
There are three possibilities for $C_\alpha$ - it can be i) non-sweepable, ii) sweepable digon cell or iii) sweepable triangle cell. 
The same three possibilities exist for $C_\beta$. We now consider all the joint possibilities for $C_\alpha$ and $C_\beta$. 
\begin{itemize}
\item {\em Both $C_\alpha$ and $C_\beta$ are non-sweepable.} In this case the new cells created after untangling $L$ are 
$C_\alpha, C_\beta$ and $D_{uv}$. By assumption, the first two are non-sweepable and by the arguments before $D_{uv}$ 
is also non-sweepable. Thus, the arrangement obtained after untangling $L$ is a simpler non-sweepable arrangement than $\AG$. 

    \item {\em One of $C_\alpha$ or $C_\beta$ is a digon cell, and the other is a non-sweepable cell.}
    Assume without loss of generality that $C_\alpha$ is a digon cell. In this case, we remove the curve $\alpha$ 
    from the original arrangement and claim
    that the resulting simpler arrangement $\mathcal{A}(\Gamma')$ is non-sweepable. The assumption that $C_\alpha$ is a digon cell implies that $D_u$ and $D_v$ are cells on $\gamma$ in the original arrangement $\Gamma$. Let $x_a$ be a reference point in $\tilde{\gamma}$ between $a_1$ and $a_2$. Let
    $C$ be the cell containing the reference point $x_\alpha$.
    We first claim that the cell $C$ is not sweepable. Note that $C$ contains the cells $D_u$ and $D_v$ from the original arrangement. If one of those cells contained $P$, then clearly $C$ is non-sweepable. Therefore, let us assume that neither $D_u$ nor $D_v$ contained $P$.
    This along with the fact that $D_u$ and $D_v$ were non-sweepable in $\Gamma$ implies that both of those cells must have at least four vertices in $\Gamma$. The cell $C$ obtained after removing $\alpha$ therefore has at least $4 + 4 - 4 = 4$ vertices. The ``$-4$" is due to the fact that we lose the four vertices $a_1, u, v$ and $a_2$ on $\alpha$ when we remove it. 
     Thus $\mathcal{A}(\Gamma')$ is a simpler non-sweepable arrangement than $\AG$ as the cell containing $u_\beta$ and $v_\beta$ remains non-sweepable. 

\begin{figure}[h!]
\begin{center}
\includegraphics[width=2in]{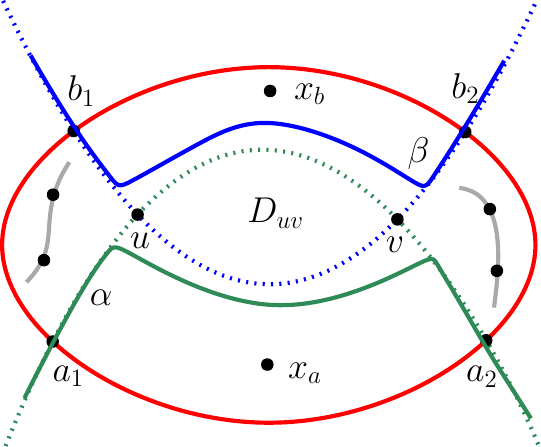}
\caption{The cells $D_u$ and $D_v$ in $\AG$ have $\ge 2$ vertices not lying on $\alpha$ or $\beta$.}
\label{fig:bothDigon}
\end{center}
\end{figure}

\item {\em Both $C_\alpha$ and $C_\beta$ are digon cells.}
In this case, we remove both $\alpha$ and $\beta$ from $\Gamma$ to obtain a new arrangement $\mathcal{A}(\Gamma')$ simpler than $\AG$. 
As in the previous case, let $x_a$ and $x_b$ be reference points in $\tilde{\gamma}$ respectively, between $a_1$ and $a_2$, and
$b_2$ and $b_1$. On removing $\alpha$ and $\beta$, there is a single cell that contains the reference points $x_a$ and $x_b$.
Let $C$ denote this cell.
We claim that $C$ is non-sweepable. This is obviously the case if $C$ contains $P$. We therefore assume that $C$ does not contain $P$. This means that in the original arrangement $\AG$, $D_u$ and $D_v$ did not contain $P$ either. Since both of these were non-sweepable cells in $\AG$, they must have at least four vertices. This along with the fact that $C_\alpha$ and $C_\beta$ are digons implies that both $D_u$ and $D_v$ had at least two vertices not lying on $\alpha$ or $\beta$. Thus, the cell $C$ created upon removing $\alpha$ and $\beta$ has at least $2+2 = 4$ vertices and is therefore non-sweepable. See Figure
~\ref{fig:bothDigon}.

\begin{figure}[ht!]
\centering
\begin{subfigure}{0.4\textwidth}
\begin{center}
 \includegraphics[scale=.6]{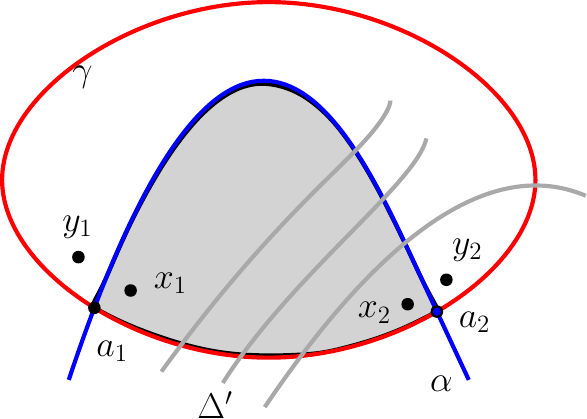}
\subcaption{\centering $\alpha$ forms a side of a triangle cell upon bypassing $L$. The grey arcs
    are the parallel arcs.}
\label{fig:triangleBeta0}
\end{center}
\end{subfigure}\hspace{1.2cm}
\begin{subfigure}{0.4\textwidth}
\begin{center}
 \includegraphics[scale=0.6]{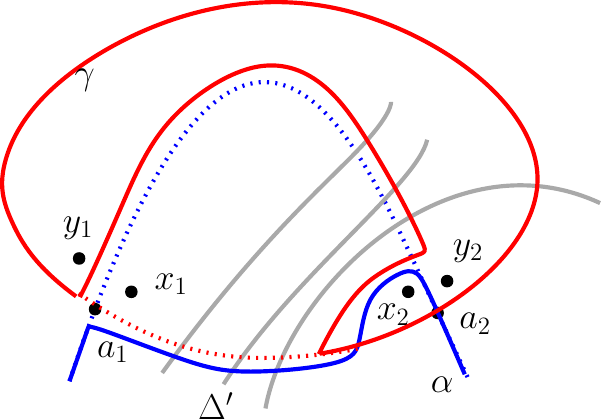}
    \subcaption{\centering Bypass the minimal triangle with base $\delta'_{\ell}$. The curves $\gamma, \alpha$ intersect outside $\delta'_{\ell}$. }
       \label{fig:triangleBeta2}
    \end{center}
 \end{subfigure}
 \caption{\centering Modifications when $C_\alpha$ is a triangle cell and $C_\beta$ is non-sweepable.}
\label{fig:minTbypassLast}
\end{figure}

    \item {\em One of $C_\alpha$ or $C_\beta$ is a triangle cell and the other is non-sweepable.} Without loss of generality, assume that $C_\alpha$ is a triangle cell and $C_\beta$ is non-sweepable. We first apply the following {\em triangle simplification} operation to $C_\alpha$. This is done as follows. Note that one of $a_1$ or $a_2$ must be a vertex of $C_\alpha$ since it is a triangle cell. Without loss of generality, we assume that it is $a_1$. Let $\delta_1$ be the curve defining the third side of $C_\alpha$ apart from $\alpha$ and $\gamma$.  We first apply Lemma~\ref{lem:parallel} with $\mu = \alpha$ and $R$ being the digon defined by $\alpha$ and $\gamma$. Let 
$\Delta_\alpha = (\delta_1, \cdots, \delta_k)$ be a maximal sequence of parallel curves in $R$ and let $\Delta'_\alpha=(\delta'_1,\ldots, \delta'_{\ell})$ denote the sequence of parallel curves obtained after the application of Lemma~\ref{lem:parallel}. By Lemma~\ref{lem:parallel}, any newly created sweepable triangle cell in the new arrangement must 
have either $a_1$ or $a_2$ as a vertex, i.e., it must be a cell containing one of the reference points $x_1, y_1, x_2$ or $y_2$ (See Figure~\ref{fig:triangleBeta0}). Note that cell containing $y_1$ is $D_{uv}$, which as argued before is not sweepable. The cell containing $x_2$ cannot be a sweepable cell since if it were, it was already a sweepable cell in $\AG$. 
The cell containing $x_1$ is a triangle cell by assumption, and it may remain a triangle cell, or become an unsweepable cell after
applying Lemma~\ref{lem:parallel}, and the cell containing
$y_2$ may have become a new triangle cell as a consequence of the modification of the parallel curves $\Delta_\alpha$
to $\Delta'_\alpha$ by Lemma~\ref{lem:parallel}.
Observe that the triangle $T$ in $R$ formed by $\alpha,\gamma$ and $\delta'_\ell$ containing $x_1$ forms a 
minimal triangle with base $\delta'_{\ell}$. 
At this point, we untangle $T$. 
By Corollary~\ref{cor:minT2int},
this yields a simpler arrangement. The cell containing $x_1$ in this new arrangement  lies outside $\tilde{\gamma}$. 
The new cell containing $y_1$ remains non-sweepable. To see this, first note that since no curve can intersect 
the segment of $\alpha$ between $u$ and $a_1$. 
Since $D_u$ is a non-sweepable cell on $\gamma$ in $\AG$, it either contains $P$ or has at least $4$ vertices. The cell $D_v$ has at least $3$ vertices since one the vertices in $v$ formed by $\alpha$ and $\beta$ and it lies outside the digon formed by $\alpha$ and $\beta$. 
Thus, the new cell $D_{uv}$ obtained upon bypassing the empty lens $L$ in $\AG$ therefore either contains $P$ or has at least $4 + 3 - 2 = 5$ vertices. The ``$-2$" is because in the process of bypassing $L$ we lose two vertices $u$ and $v$. 
Now, the cell containing $y_1$ after untangling the triangle $T$ contains the former cell $D_{uv}$ and one less side than $D_{uv}$. This means that it either contains $P$ or has at least 4 sides and is therefore non-sweepable. 
Thus the arrangement obtained after bypassing $T$ is a simpler non-sweepable arrangement than $\AG$. See Figure~\ref{fig:minTbypassLast}.

\begin{figure}[ht!]
\begin{center}
    \includegraphics[width=3in]{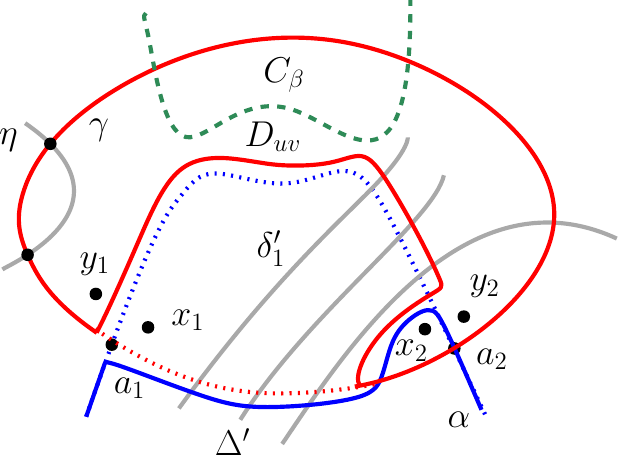}
\caption{If $C_\beta$ is a digon cell and $C_\alpha$ is a triangle cell, then on removing $C_\beta$, and applying triangle
simplification followed by triangle bypassing on $C_\alpha$ we obtain a simpler non-sweepable arrangement. In particular, the cell $D_{uv}$ has at least 4 sides
as $\gamma$ contributes at least 2 disjoint sides to $D_{uv}$.}
\label{fig:oneDigonOneTriangle}
\end{center}
\end{figure}
\item {\em One of $C_\alpha$ and $C_\beta$ is a digon cell and the other is a triangle cell.}
Without loss of generality, we assume that $C_\alpha$ is a triangle cell and $C_\beta$ is a digon cell, and that one vertex of the triangle $C_\alpha$ is $a_1$. 
Since $D_u$ is a not a bypassable cell in $\AG$, it either contains $P$ or there is a curve $\eta$ which intersects the portion of $\gamma$ between $b_1$ and $a_1$ twice. See Figure~\ref{fig:oneDigonOneTriangle} for the latter case. We now remove the curve $\beta$ and apply triangle simplification (introduced in the previous case) to $C_\alpha$. 
We 
claim that the new arrangement, which is simpler, remains non-sweepable. 
As before, let $R$ be the digon defined by $\alpha$ and $\gamma$ and let  $\Delta'_\alpha=(\delta_1', \cdots, \delta'_\ell)$ be the sequence of parallel curves obtained in $R$ as a result of the triangle simplification. From the earlier arguments, the only cell which may potentially have become bypassable is the cell containing the reference point $y_1$. 
If the cell $D_u$ in $\Gamma$ contained the point $P$, this cell also contains $P$ and is therefore non-bypassable. Otherwise, there is a curve $\eta$ as claimed before and this implies that $\gamma$ contributes two disjoint sides to the cell containing $y_1$. The cell thus has at least four sides as is therefore not bypassable.

\begin{figure}[ht!]
\centering
\begin{subfigure}{0.4\textwidth}
\begin{center}
 \includegraphics[scale=0.6]{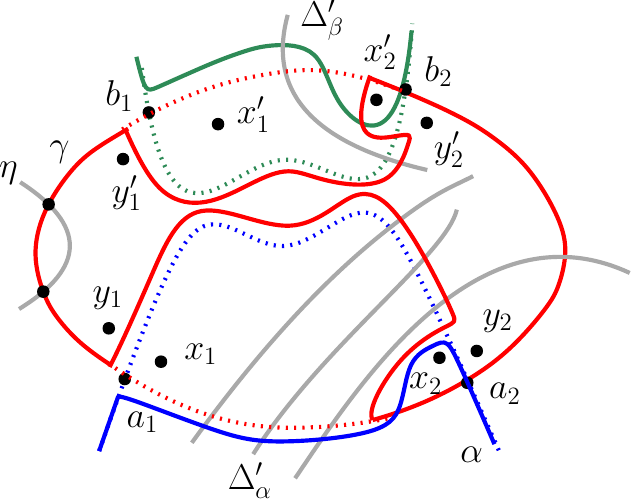}
    \subcaption{\centering Vertices $a_1$ and $b_2$ respectively contributing a vertex of the triangle cells.}
       \label{fig:bothTrianglescase1}
    \end{center}
 \end{subfigure}
 \hspace{1.2cm}
 \begin{subfigure}{0.4\textwidth}
\begin{center}
 \includegraphics[scale=.6]{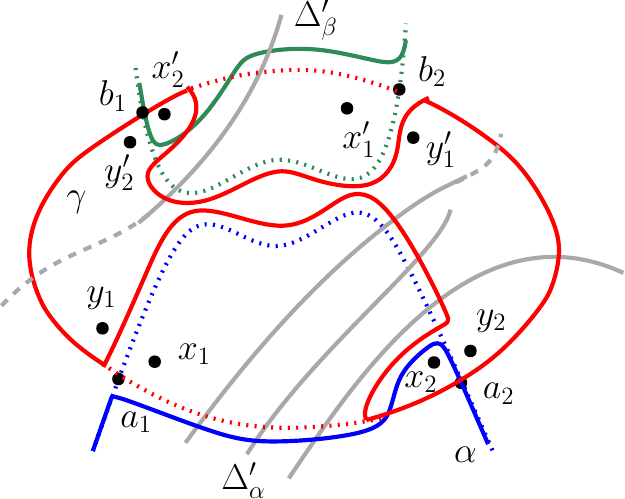}
\subcaption{\centering  Vertices $a_1$ and $b_1$ respectively contributing a vertex of the triangle cells.}
\label{fig:bothTrianglescase2}
\end{center}
\end{subfigure}
 \caption{\centering The cells $C_\alpha$ and $C_\beta$ are both triangle cells.}
 \label{fig:bothTriangles}
 \end{figure}

\item {\em Both $C_\alpha$ and $C_\beta$ are triangle cells.}
In this case, we apply triangle simplification to both triangles. Let $\Delta'_\alpha = (\delta'_1, \cdots, \delta'_\ell)$ be the sequence of parallel curves obtained as a result of applying triangle simplification to $C_\alpha$ and let $x_1, y_1, x_2, y_2$ be the reference points. Similarly let $\Delta'_\beta = (\delta''_1, \cdots, \delta''_k)$ be the sequence of parallel curves obtained as a result of bypassing $C_\beta$ and 
let $x'_1, y'_1, x'_2, y'_2$ be the reference points.
By the earlier arguments, the only cell that can potentially become sweepable are those that contain $y_1$ or $y'_1$. We claim that neither of these are sweepable.
Without loss of generality, assume that $a_1$ is vertex of $C_\alpha$. Now, either $b_1$ or $b_2$ is a vertex of $C_\beta$. We treat these two cases separately. First consider the case in which $b_1$ is a vertex of $C_\beta$. Since $D_u$ was not bypassable in $\AG$, it either contains the point $P$ or there is a curve $\eta$ that intersects the portion of $\gamma$ between $b_1$ and $a_1$ twice. The situation in $\AG$ as well as the situation after triangle simplification is applied to $C_\alpha$ and $C_\beta$ are shown in  Figure~\ref{fig:bothTrianglescase1}.
As the Figure shows, $y_1$ and $y'_1$ are contained in the same cell.
Furthermore, the cell either contains $P$ or the curve $\eta$ as mentioned above exists, in which case $\gamma$ contributes two disjoint sides to it and the cell has at least four sides. In either case, the cell is non-sweepable. 
Now consider the other case where $b_2$ is a vertex of $C_\beta$.  The situation for this case in $\AG$ as well as the situation after triangle simplification is applied to $C_\alpha$ and $C_\beta$ are shown in  Figure~\ref{fig:bothTrianglescase2}. As the figure shows, the cells containing $y_1$ and $y'_1$ are again identical. Furthermore, the cell either contains $P$ or $\gamma$ contributes two disjoint sides to it. In either case, the cell is non-sweepable.
This completes the proof. 
\end{itemize}

\end{proof}

\section{Applications}
\label{sec:applications}

In this section, we describe several applications of Theorem~\ref{thm:mainthm}. We omit the proofs of some of the theorems in this section since they only require replacing the result of Snoeyink and Hershberger (Theorem~\ref{thm:snoeyink}) with our result (Theorem~\ref{thm:mainthm}). 

\smallskip\noindent
{\bf Every arrangement of non-piercing regions contains a ``small'' non-piercing region:}
Pinchasi~\cite{DBLP:journals/siamdm/Pinchasi14} proved that an arrangement of pseudodisks $\Gamma$ contains a pseudodisk $\gamma$
s.t. $\gamma$ is intersected by at most 156 disjoint pseudodisks in $\Gamma\,\setminus\,\{\gamma\}$. 
Using Theorem~\ref{thm:mainthm},
Pinchasi's result extends immediately to non-piercing regions.

\begin{theorem}
\label{thm:small}
Let $\AG$ be an arrangement of non-piercing curves $\Gamma$ in the plane. Then, there is a curve $\gamma\in\Gamma$ s.t.
the number of curves in $\Gamma$ whose regions are disjoint and intersect $\tilde{\gamma}$ is at most 156.
\end{theorem}

\hide{
A classic algorithm to compute an independent set in an arrangement $\mathcal{D}$ of disks in the plane is to greedily select a smallest disk and remove the disks intersecting it.
For any arrangement of disks, it is easy to show that there is a disk s.t. there are at most 5 disjoint disks intersecting it. In particular, the disk with smallest
radius in $\mathcal{D}$ has this property. This implies that the greedy algorithm described yields a $5$-approximation for the Maximum Independent Set problem 
in the intersection graph of disks in the plane. From Theorem~\ref{thm:small}, it follows by the same argument that a greedy algorithm yields a 156-approximation for 
the Maximum Independent Set problem in the intersection graph of non-piercing regions in the plane.}

\smallskip
\noindent
{\bf Multi-Hitting set with non-piercing regions:}
Raman and Ray~\cite{DBLP:journals/dcg/RamanR22} 
studied the \emph{Multi-Hitting Set} and \emph{Multi-Set Cover} problems for set systems defined by a set of points and an arrangement of non-piercing regions in the plane. In the Multi-Hitting Set problem, the input is a set $P$ of points and a set $\mathcal{D}$ of non-piercing regions, and a demand
function $d:\mathcal{D}\to\mathbb{N}$. A feasible solution is a set $P'\subseteq P$ such that $|D\cap P'|\ge d(D)$ for each $D\in\mathcal{D}$.
In the Multi-Set Cover problem, the input is again a set $P$ of points and set $\mathcal{D}$ of non-piercing regions. The demand function 
in this setting is on the points. Thus $d:P\to\mathbb{N}$ is the demand function. A feasible solution is a set $\mathcal{D}'\subseteq\mathcal{D}$
such that $|\{D\in\mathcal{D}': p\in D\}|\ge d(p)$  for each $p\in P$. The optimization problem in both cases is to find a feasible solution
of smallest cardinality. 

Raman and Ray~\cite{DBLP:journals/dcg/RamanR22} showed that there is a PTAS for the Multi-Set Cover problem if the maximum demand on the points is bounded by a
constant. On the other hand, if the demands are arbitrary, they showed a $(2+\epsilon)$-
approximation by a combination of linear programming rounding and local search. In particular, they showed that the \emph{quasi-uniform}
sampling technique~\cite{DBLP:conf/soda/ChanGKS12,DBLP:conf/stoc/Varadarajan10} can be used to obtain a fractional solution with bounded depth without increasing the LP value, and then 
we can switch to a local search algorithm. 
The analysis of the local search algorithm requires the construction of a suitable graph. For the Multi-Hitting Set version of the
problems, the authors used the sweeping theorem of Snoeyink and Hershberger~\cite{SnoeyinkH90}, and hence their results only applied to
the Multi-Hitting Set problem with points and pseudodisks. However, with Theorem~\ref{thm:mainthm} at hand, we can extend their result
to obtain matching results for the Multi-Hitting Set problem of points and non-piercing regions, as with the Multi-Set Cover problem
with points and non-piercing regions.

\begin{theorem}
\label{thm:multihitting}
Given an instance $(P,\mathcal{S})$ where $P$ is a set of points in the plane and $\mathcal{S}$ is an arrangement of non-piercing regions
in the plane with demands $d:\mathcal{S}\to\mathbb{N}$, there is a PTAS for the Multi-Hitting set problem, i.e., selecting the smallest
size subset $Q\subseteq P$ s.t. $\forall S\in\mathcal{S}$, $Q\cap S\ge d(S)$ when the demands are 
bounded above by a constant, and otherwise there is a $(2+\epsilon)$-approximation algorithm.
\end{theorem}

\smallskip\noindent 
{\bf Number of hyperedges defined by lines and non-piercing regions:}
Keller et al.,~\cite{DBLP:journals/combinatorics/KellerKP22} considered the hypergraph $(\mathcal{L},\mathcal{D})$, where $\mathcal{L}$ is a set of lines and $\mathcal{D}$ is 
set of pseudodisks in the plane. Each $D \in \mathcal{D}$
defines a hyperedge consisting of the lines in $\mathcal{L}$ that intersect $D$. The authors prove that
the number of hyperedges in such a hypergraph is $O_t(|\mathcal{L}|^3)$. 
By using Theorem~\ref{thm:mainthm} in place of the result of Snoeyink and Hershberger (Theorem~\ref{thm:snoeyink}), and following the proof in~\cite{DBLP:journals/combinatorics/KellerKP22}, we obtain the same bound for the hypergraph defined by lines and non-piercing regions.

\begin{theorem}
\label{thm:keller}
Let $(\mathcal{L},\mathcal{D})$ be a hypergraph defined by a set $\mathcal{L}$ of lines and a set $\mathcal{D}$ of non-piercing regions in the plane, where each $D \in \mathcal{D}$ is defines the hyperedge consisting of the set of lines  in $\mathcal{L}$ that intersect $D$. Then, the number of hyperedges of size $t$ in this hypergraph is $O_t(|\mathcal{L}|^2)$,
and the total number of hyperedges is $O(|\mathcal{L}|^3)$. Both bounds are tight already when $\mathcal{D}$ is a set of pseudodisks.
\end{theorem}

\smallskip
\noindent
{\bf Construction of Supports:}
Raman and Ray~\cite{DBLP:journals/dcg/RamanR20} proved the following result. Using Theorem~\ref{thm:mainthm} and Lemma~\ref{lem:digonbypass}, we obtain a significantly simpler proof. 
\begin{restatable}[Support for non-piercing regions]{theorem}{support}
\label{thm:support}
Let $\mathcal{H}$ be an arrangement of non-piercing regions and let $P$
be any set of points in the plane. 
Then, there exists a planar support for $\mathcal{H}$ i.e., a planar graph $G=(P,E)$ s.t. for any $H \in \mathcal{H}$ the subgraph of $G$ induced by $H \cap P$ is connected.
\end{restatable}

\begin{proof}
We first consider the special case where each region contains exactly two points. 
In this case, it suffices to set $E = \{(p, q) : \exists \text{ a region } H \in \mathcal{H} \text{ containing } p \text{ and } q\}$.  To show that this graph is planar, we first obtain a drawing of the graph as follows. Each vertex $p \in P$ is represented by the point $p$ itself and for each edge $(p, q)$ where $p$ and $q$ belong to a region $R$, we construct a curve joining $p$ and $q$ whose interior lies in the interior of $R$. 
By Lemma 1 of~\cite{buzaglo2013topological}, the curves corresponding to every pair of non-adjacent edges (i.e., edges not sharing a vertex) intersect an even number of times. Thus, by the strong Hanani-Tutte theorem~\cite{schaefer2013hanani}, the graph is planar. 

Next, we reduce the general case to the special case above.
We apply sweeping to each region $H\in\mathcal{H}$ and add an additional set of regions $\mathcal{K}$, each
of which contains at most two points of $P$, and such that $\mathcal{H}\cup\mathcal{K}$ is non-piercing.
Initially, $\mathcal{K}$ consists of a set of disjoint regions, each of which is a small region containing
a single point of $P$ so that $\mathcal{H}\cup\mathcal{K}$ is an arrangement of non-piercing regions.
For a region $H\in\mathcal{H}$, let $\mathcal{K}_H$ be the set of regions of $\mathcal{K}$ contained in $H$.

Choose an $H\in\mathcal{H}$ such that the union of regions $\mathcal{K}_H$, i.e., regions of $\mathcal{K}$ contained
in $H$ form more than one path connected component. 
Make a copy $H'$ of $H$ and sweep $H'$ in the interior of $H$ until there
is a component of $\mathcal{K}_H$ that shares exactly one point, say $p$ with $H'$. At this point, observe that $H'$
still intersects more than one connected component of $\mathcal{K}_H$. Now, we sweep $H'$ to $p$, and stop
when there are at most two points in $H'$: $p$ and a point $q$ that lies in a different component of $\mathcal{H}_K$,
and applying any other sweep operation will make $q$ lie outside $H'$. Let $K=H'$ be the new region added to 
$\mathcal{K}$, and remove from $\mathcal{K}$, 
any region of $\mathcal{K}$ containing only one point that lies in $K$. By Theorem~\ref{thm:mainthm}, the
regions $\mathcal{H}\cup\mathcal{K}$ are non-piercing.
We repeat this process until the union of regions in $\mathcal{K}_H$
form a path connected region for every $H \in \mathcal{H}$.
At the end of this process, we have a set of regions $\mathcal{K}$,
each of which contains exactly two points, and a support for $\mathcal{K}$ is a support for $\mathcal{H}$.  This is the special case we considered earlier.  Figure~\ref{fig:sweepSupport} shows an example.
\hide{Next, we simplify the arrangement of $\mathcal{K}$ as follows: We say that a digon in the arrangement
of $\mathcal{K}$ is empty if it does not contain a point of $P$. As long as there is a minimal empty
digon in the arrangement of $\mathcal{K}$ we can simplify the arrangement by applying digon untangling. By Lemma~\ref{lem:digonbypass},
the arrangement of the regions in $\mathcal{K}$ remain non-piercing.
At the end of this operation, the only regions of $\mathcal{K}$ that intersect are those that share a point of $P$.
For each point $p\in P$, and $K\in\mathcal{K}$ such that $p\in K$, we connect $p$ to the other point
of $P$ in $K$ via a continuous curve lying in $K$. We can do this in such a way that the curves from
$p$ are internally disjoint. This yields the desired planar support. Figure~\ref{fig:sweepSupport} shows the 
construction of a support.}
\end{proof}
\begin{figure}[!h]
\begin{center}
\includegraphics[width=3in]{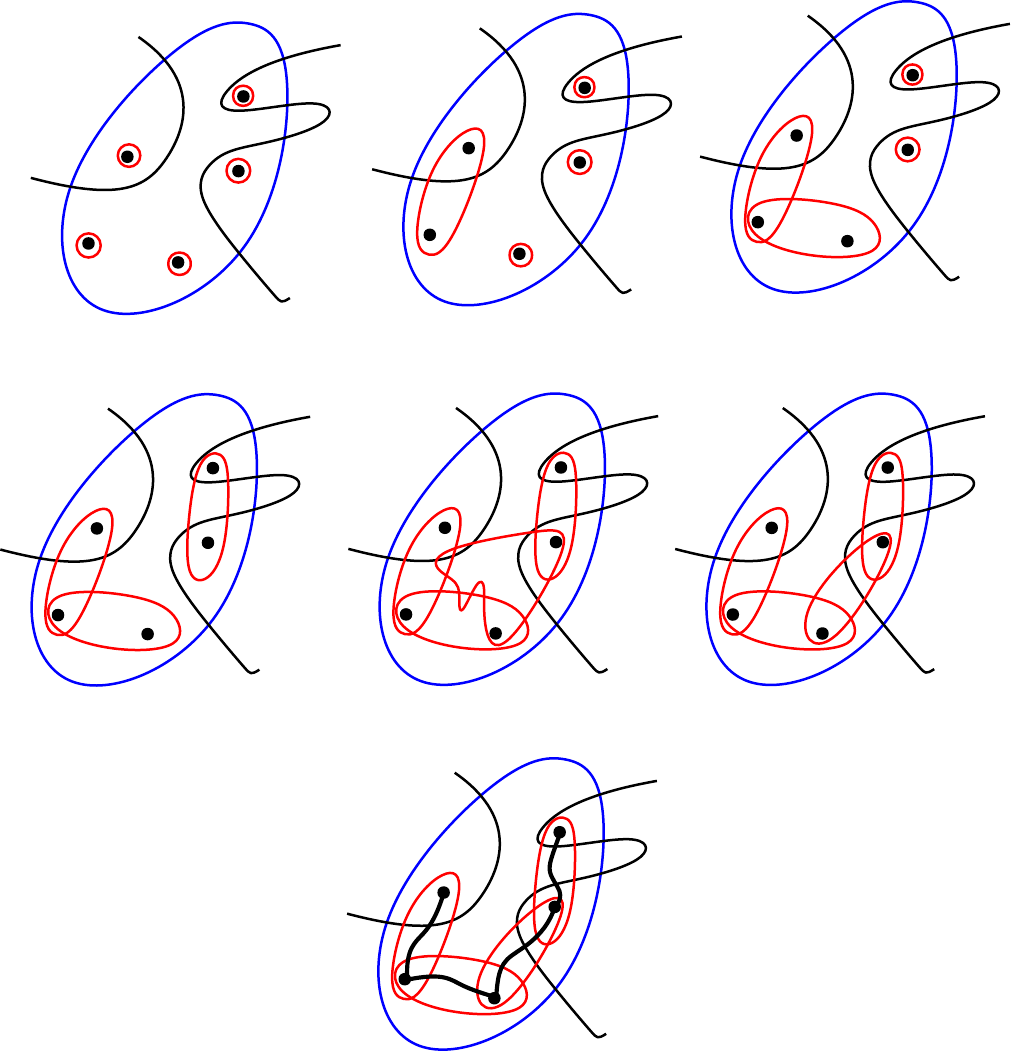}
\caption{The figure above shows the region $H$ in blue, and the regions $\mathcal{K}$ in red obtained by sweeping a copy of $H$
until it contains exactly two points from distinct connected components of $\mathcal{K}$. 
The last figure shows the curves between the points in each region in $K$ so that the resulting graph
is a plane graph.}
\label{fig:sweepSupport}
\end{center}
\end{figure}
\section{Conclusion}
\label{sec:conclusion}
We gave an alternate proof of the result of Snoeyink and Hershberger~\cite{SnoeyinkH90}, in the process developing conceptual tools for modifying arrangements of regions. We also extended their result to non-piercing regions. For both pseudocircles and for non-piercing regions, implementing the sweep in an algorithmically efficient manner remains an open problem.

Snoeyink and Hershberger showed that their result does not extend to $k$-intersecting arrangements of (open or closed) curves for $k \geq 3$. A natural question is whether it can be extended for $k$-intersecting non-piercing arrangements of curves (also called $k$-admissible in the literature~\cite{MR10}).  Unfortunately, this is not possible either -  Figure~\ref{fig:k_admissible} shows a counter-example with $5$ curves which are $4$-admissible but any process of shrinking the sweep curve $\gamma$ results in $\gamma$ intersecting one of the curves in $6$ points in between. We believe however, that if we are given a set of $k$ intersecting non-piercing curves, it is possible to sweep with one of the curves so that the arrangement
remains at most $2k$ intersecting. We leave this as an open question.

Another basic result in the paper of Snoeyink and Hershberger is an extension of the Levi extension theorem. In particular, the authors show that for any set of two-intersecting curves
and 2 points not on the same curve, we can draw another curve through the two points so that the arrangement remains two-intersecting.
This theorem also extends to non-piercing regions. In particular, given a set of non-piercing curves and 3 points, 
we can add a new curve that passes through the 3 points and the arrangement remains non-piercing. 
This result will appear in a companion paper.
\begin{figure}[ht!]
\begin{center}
  \includegraphics[scale=.4]{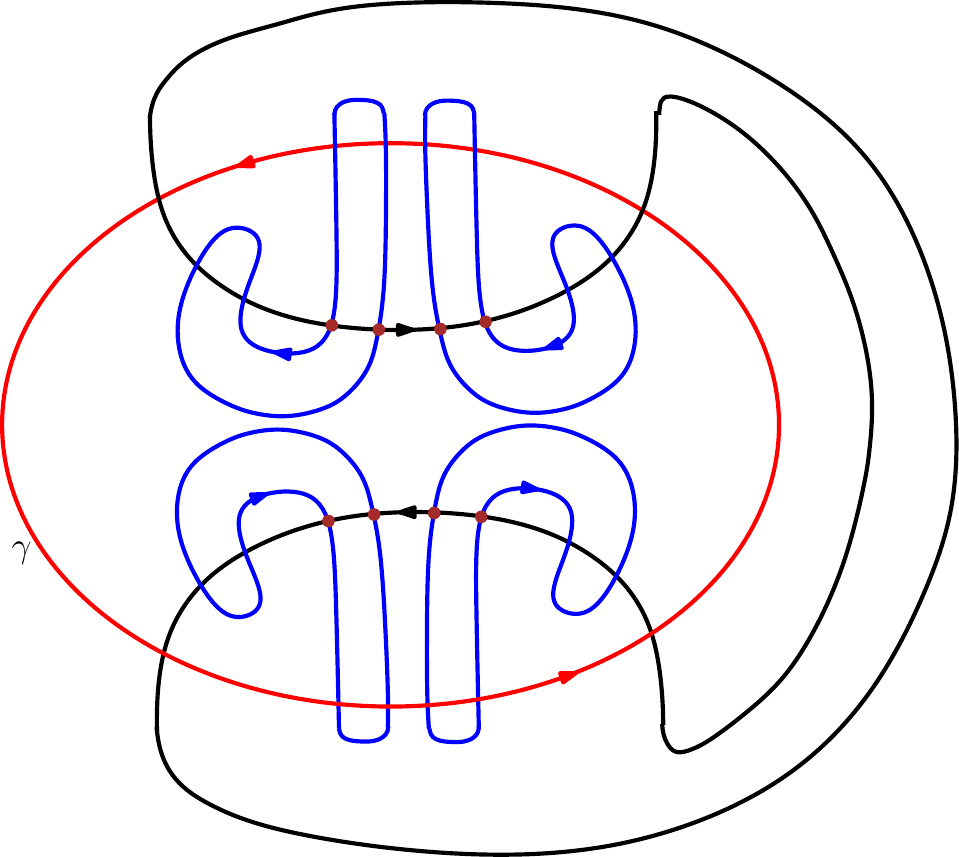}
\end{center}
 \caption{Arrangement of $4$-intersecting non-piercing curves ($4$-admissible curves) that can not be swept with the curve $\gamma$ maintaining $4$-admissibility.}
 \label{fig:k_admissible}
 \end{figure}
\bibliography{ref}

@article{DBLP:journals/dcg/RamanR22,
  author       = {Rajiv Raman and
                  Saurabh Ray},
  title        = {On the Geometric Set Multicover Problem},
  journal      = {Discrete \& Computational Geometry},
  volume       = {68},
  number       = {2},
  pages        = {566--591},
  year         = {2022},
  url          = {https://doi.org/10.1007/s00454-022-00402-y},
  doi          = {10.1007/S00454-022-00402-Y},
  bibsource    = {dblp computer science bibliography, https://dblp.org}
}

@incollection{schaefer2013hanani,
  title={Hanani-Tutte and related results},
  author={Schaefer, Marcus},
  booktitle={Geometry—Intuitive, Discrete, and Convex: A Tribute to L{\'a}szl{\'o} Fejes T{\'o}th},
  pages={259--299},
  year={2013},
  publisher={Springer}
}

@article{DBLP:journals/dcg/RoyGRR18,
  author       = {Aniket Basu Roy and
                  Sathish Govindarajan and
                  Rajiv Raman and
                  Saurabh Ray},
  title        = {Packing and Covering with Non-Piercing Regions},
  journal      = {Discrete \& Computational Geometry},
  volume       = {60},
  number       = {2},
  pages        = {471--492},
  year         = {2018},
  url          = {https://doi.org/10.1007/s00454-018-9983-2},
  doi          = {10.1007/S00454-018-9983-2},
  bibsource    = {dblp computer science bibliography, https://dblp.org}
}

@misc{keller2026,
      title={New Sufficient Conditions for Linear-Sized Epsilon-Nets and $(p,2)$-Theorems}, 
      author={Chaya Keller and Shakhar Smorodinsky},
      year={2026},
      eprint={2507.07269},
      archivePrefix={arXiv},
      primaryClass={math.CO},
      url={https://arxiv.org/abs/2507.07269}, 
}

@article{DBLP:journals/siamdm/Pinchasi14,
  author       = {Rom Pinchasi},
  title        = {A Finite Family of Pseudodiscs Must Include a "Small" Pseudodisc},
  journal      = {{SIAM} Journal of Discrete Mathematics},
  volume       = {28},
  number       = {4},
  pages        = {1930--1934},
  year         = {2014},
  url          = {https://doi.org/10.1137/130949750},
  doi          = {10.1137/130949750},
  bibsource    = {dblp computer science bibliography, https://dblp.org}
}

@article{DBLP:journals/combinatorics/KellerKP22,
  author       = {Chaya Keller and
                  Bal{\'{a}}zs Keszegh and
                  D{\"{o}}m{\"{o}}t{\"{o}}r P{\'{a}}lv{\"{o}}lgyi},
  title        = {On the Number of Hyperedges in the Hypergraph of Lines and Pseudo-Discs},
  journal      = {Electronic Journal of Combinatorics},
  volume       = {29},
  number       = {3},
  year         = {2022},
  url          = {https://doi.org/10.37236/10424},
  doi          = {10.37236/10424},
  bibsource    = {dblp computer science bibliography, https://dblp.org}
}

@article{DBLP:journals/dcg/RamanR20,
  author       = {Rajiv Raman and
                  Saurabh Ray},
  title        = {Constructing Planar Support for Non-Piercing Regions},
  journal      = {Discrete \& Computational Geometry},
  volume       = {64},
  number       = {3},
  pages        = {1098--1122},
  year         = {2020},
  url          = {https://doi.org/10.1007/s00454-020-00216-w},
  doi          = {10.1007/S00454-020-00216-W},
  bibsource    = {dblp computer science bibliography, https://dblp.org}
}

@book{de2000computational,
  author       = {Mark de Berg and
                  Otfried Cheong and
                  Marc J. van Kreveld and
                  Mark H. Overmars},
  title        = {Computational Geometry: Algorithms and Applications, 3rd Edition},
  publisher    = {Springer},
  year         = {2008},
  url          = {https://www.worldcat.org/oclc/227584184},
  isbn         = {9783540779735},
  bibsource    = {dblp computer science bibliography, https://dblp.org}
}

@article{edelsbrunner1986topologically,
title = {Topologically sweeping an arrangement},
author = {Herbert Edelsbrunner and Leonidas J. Guibas},
journal = {Journal of Computer and System Sciences},
volume = {38},
number = {1},
pages = {165-194},
year = {1989},
issn = {0022-0000},
doi = {https://doi.org/10.1016/0022-0000(89)90038-X},
url = {https://www.sciencedirect.com/science/article/pii/002200008990038X},
}

@article{felsner1999triangles,
  title={Triangles in Euclidean arrangements},
  author={Felsner, Stefan and Kriegel, Klaus},
  journal={Discrete \& Computational Geometry},
  volume={22},
  pages={429--438},
  year={1999},
  publisher={Springer}
}

@article{ackerman2021coloring,
  title={Coloring {D}elaunay-edges and their generalizations},
  author={Ackerman, Eyal and Keszegh, Bal{\'a}zs and P{\'a}lv{\"o}lgyi, D{\"o}m{\"o}t{\"o}r},
  journal={Computational Geometry},
  volume={96},
  pages={101745},
  year={2021},
  publisher={Elsevier}
}

@inproceedings{buzaglo2008s,
  title={On s-intersecting curves and related problems},
  author={Buzaglo, Sarit and Holzman, Rom and Pinchasi, Rom},
  booktitle={Proceedings of the twenty-fourth annual Symposium on Computational geometry},
  pages={79--84},
  year={2008}
}

@article{chiu2023coloring,
  title={Coloring circle arrangements: New 4-chromatic planar graphs},
  author={Chiu, Man-Kwun and Felsner, Stefan and Scheucher, Manfred and Schr{\"o}der, Felix and Steiner, Raphael and Vogtenhuber, Birgit},
  journal={European Journal of Combinatorics},
  pages={103839},
  year={2023},
  publisher={Elsevier}
}

@article{agarwal2004lenses,
  title={Lenses in arrangements of pseudo-circles and their applications},
  author={Agarwal, Pankaj K. and Nevo, Eran and Pach, J{\'a}nos and Pinchasi, Rom and Sharir, Micha and Smorodinsky, Shakhar},
  journal={Journal of the ACM (JACM)},
  volume={51},
  number={2},
  pages={139--186},
  year={2004},
  publisher={ACM New York, NY, USA}
}

@phdthesis{scheucher2020points,
  title={Points, Lines, and Circles: Some Contributions to Combinatorial Geometry},
  author={Scheucher, Manfred},
  year={2020},
  school={Dissertation, Berlin, Technische Universit{\"a}t Berlin, 2019}
}

@article{chazelle1992optimal,
  title={An optimal algorithm for intersecting line segments in the plane},
  author={Chazelle, Bernard and Edelsbrunner, Herbert},
  journal={Journal of the ACM (JACM)},
  volume={39},
  number={1},
  pages={1--54},
  year={1992},
  publisher={ACM New York, NY, USA}
}

@article{bentley1979algorithms,
  title={Algorithms for reporting and counting geometric intersections},
  author={Bentley, Jon L. and Ottmann, Thomas A.},
  journal={IEEE Transactions on computers},
  volume={100},
  number={9},
  pages={643--647},
  year={1979},
  publisher={IEEE}
}

@article{DBLP:journals/comgeo/ChanG14,
  author       = {Timothy M. Chan and
                  Elyot Grant},
  title        = {Exact algorithms and APX-hardness results for geometric packing and
                  covering problems},
  journal      = {Computational Geometry},
  volume       = {47},
  number       = {2},
  pages        = {112--124},
  year         = {2014},
  url          = {https://doi.org/10.1016/j.comgeo.2012.04.001},
  doi          = {10.1016/J.COMGEO.2012.04.001},
  bibsource    = {dblp computer science bibliography, https://dblp.org}
}

@article{DBLP:journals/siamcomp/Har-PeledQ17,
  author       = {Sariel Har{-}Peled and
                  Kent Quanrud},
  title        = {Approximation Algorithms for Polynomial-Expansion and Low-Density
                  Graphs},
  journal      = {{SIAM} Journal on Computing},
  volume       = {46},
  number       = {6},
  pages        = {1712--1744},
  year         = {2017},
  url          = {https://doi.org/10.1137/16M1079336},
  doi          = {10.1137/16M1079336},
  bibsource    = {dblp computer science bibliography, https://dblp.org}
}

@book{preparata2012computational,
  title={Computational Geometry: An Introduction},
  author={Preparata, Franco P. and Shamos, Michael I.},
  year={2012},
  publisher={Springer Science \& Business Media}
}

@inproceedings{DBLP:conf/stoc/Varadarajan10,
  author       = {Kasturi R. Varadarajan},
  editor       = {Leonard J. Schulman},
  title        = {Weighted geometric set cover via quasi-uniform sampling},
  booktitle    = {Proceedings of the 42nd {ACM} Symposium on Theory of Computing, {STOC}
                  2010, Cambridge, Massachusetts, USA, 5-8 June 2010},
  pages        = {641--648},
  publisher    = {{ACM}},
  year         = {2010},
  url          = {https://doi.org/10.1145/1806689.1806777},
  doi          = {10.1145/1806689.1806777},
  bibsource    = {dblp computer science bibliography, https://dblp.org}
}

@inproceedings{DBLP:conf/soda/ChanGKS12,
  author       = {Timothy M. Chan and
                  Elyot Grant and
                  Jochen K{\"{o}}nemann and
                  Malcolm Sharpe},
  editor       = {Yuval Rabani},
  title        = {Weighted capacitated, priority, and geometric set cover via improved
                  quasi-uniform sampling},
  booktitle    = {Proceedings of the Twenty-Third Annual {ACM-SIAM} Symposium on Discrete
                  Algorithms, {SODA} 2012, Kyoto, Japan, January 17-19, 2012},
  pages        = {1576--1585},
  publisher    = {{SIAM}},
  year         = {2012},
  url          = {https://doi.org/10.1137/1.9781611973099.125},
  doi          = {10.1137/1.9781611973099.125},
  bibsource    = {dblp computer science bibliography, https://dblp.org}
}

@book{edelsbrunner1987algorithms,
  title={Algorithms in Combinatorial Geometry},
  author={Edelsbrunner, Herbert},
  volume={10},
  year={1987},
  publisher={Springer Science \& Business Media}
}

@book{DBLP:books/daglib/0030489,
  author       = {Bojan Mohar and
                  Carsten Thomassen},
  title        = {Graphs on Surfaces},
  series       = {Johns Hopkins series in the mathematical sciences},
  publisher    = {Johns Hopkins University Press},
  year         = {2001},
  url          = {http://jhupbooks.press.jhu.edu/ecom/MasterServlet/GetItemDetailsHandler?iN=9780801866890\&qty=1\&source=2\&viewMode=3\&loggedIN=false\&JavaScript=y},
  isbn         = {978-0-8018-6689-0},
  bibsource    = {dblp computer science bibliography, https://dblp.org}
}

@article{buzaglo2013topological,
  title={Topological hypergraphs},
  author={Buzaglo, Sarit and Pinchasi, Rom and Rote, G{\"u}nter},
  journal={Thirty Essays on Geometric Graph Theory},
  pages={71--81},
  year={2013},
  publisher={Springer}
}

@incollection{SnoeyinkH90,
  author       = {Jack Snoeyink and
                  John Hershberger},
  editor       = {Jacob E. Goodman and
                  Richard Pollack and
                  William Steiger},
  title        = {Sweeping Arrangements of Curves},
  booktitle    = {Discrete and Computational Geometry: Papers from the {DIMACS} Special  Year},
  series       = {{DIMACS} Series in Discrete Mathematics and Theoretical Computer Science},
  volume       = {6},
  pages        = {309--350},
  publisher    = {{DIMACS/AMS}},
  year         = {1990},
  doi          = {10.1090/DIMACS/006/21},
  }

@article{MR10,
  author       = {Nabil H. Mustafa and
                  Saurabh Ray},
  title        = {Improved Results on Geometric Hitting Set Problems},
  journal      = {Discrete \& Computational Geometry},
  volume       = {44},
  number       = {4},
  pages        = {883--895},
  year         = {2010},
  url          = {https://doi.org/10.1007/s00454-010-9285-9},
  doi          = {10.1007/S00454-010-9285-9},
  bibsource    = {dblp computer science bibliography, https://dblp.org}
}

\end{document}